\documentclass[journal]{IEEEtran}
\ifCLASSINFOpdf
\else
\fi
\usepackage{amsmath}
\usepackage{amssymb}
\usepackage{amsthm}
\usepackage{stfloats}
\usepackage{algorithm}
\usepackage{algorithmic}
\usepackage{multirow}
\usepackage{amsmath}
\usepackage{color}
\usepackage{xcolor}
\usepackage{cite}
\usepackage{amsfonts,amssymb}
\usepackage{dsfont}
\usepackage{mathrsfs}
\usepackage{graphicx}
\usepackage{subfigure}
\usepackage{longtable}
\usepackage{stfloats}
\usepackage{makecell}

\newtheorem{proposition}{Proposition}

\begin{document}

\title{\huge Dynamic UAV Swarm Collaboration for Multi-Targets Tracking under Malicious Jamming: Joint Power, Path and Target Association Optimization}
%
%
%

\author{Lanhua~Xiang,
        Fengyu~Wang,~\IEEEmembership{Member,~IEEE,}
        Wenjun~Xu,~\IEEEmembership{Senior Member,~IEEE,}
        Tiankui~Zhang,~\IEEEmembership{Senior Member,~IEEE,}
        Miao~Pan,~\IEEEmembership{Senior Member,~IEEE,}
        and~Zhu~Han,~\IEEEmembership{Fellow,~IEEE}
\thanks{Lanhua Xiang and Fengyu Wang are with the School of Artificial Intelligence, Beijing University of Posts and Telecommunications, Beijing 100876, China.}
\thanks{Wenjun Xu is with the Key Laboratory of Universal Wireless Communications, Ministry of Education, School of Artificial Intelligence, Beijing University of Posts and Telecommunications, Beijing 100876, China,
and also with the Frontier Research Center, Peng Cheng Laboratory, Shenzhen 518066, China. (e-mail: wjxu@bupt.edu.cn)}
\thanks{Tiankui Zhang is with the School of Information and Communication Engineering, Beijing University of Posts and Telecommunications, Beijing 100876, China. (e-mail: zhangtiankui@bupt.edu.cn)}
\thanks{Miao Pan is  with the Department of Electrical and Computer Engineering, University of Houston, Houston, TX 77004 USA.}
\thanks{Zhu Han is with the Department of Electrical and Computer Engineering, University of Houston, Houston, TX 77004 USA, and also with the Department of Computer Science and Engineering, Kyung Hee University, Seoul 446-701, South Korea.}
\thanks{Corresponding authors: Wenjun Xu.}        }
\maketitle

\begin{abstract}
In this paper, the multi-target tracking (MTT) with an unmanned aerial vehicle (UAV) swarm is investigated in the presence of jammers, where UAVs in the swarm communicate with each other to exchange information of targets during tracking. The communication between UAVs suffers from severe interference, including inter-UAV interference and jamming, thus leading to a deteriorated quality of MTT. To mitigate the interference and achieve MTT, we formulate a interference minimization problem by jointly optimizing UAV's sub-swarm division, trajectory, and power, subject to the constraint of MTT, collision prevention, flying ability, and UAV energy consumption. Due to the multiple coupling of sub-swarm division, trajectory, and power, the proposed optimization problem is NP-hard. To solve this challenging problem, it is decomposed into three subproblems, i.e., target association, path plan, and power control. First, a cluster-evolutionary target association (CETA) algorithm is proposed, which involves dividing the UAV swarm into the multiple sub-swarms and individually matching these sub-swarms to targets. Second, a jamming-sensitive and singular case tolerance (JSSCT)-artificial potential field (APF) algorithm is proposed to plan trajectory for tracking the targets. Third, we develop a jamming-aware mean field game (JA-MFG) power control scheme, where a novel cost function is established considering the total interference. Finally, to minimize the total interference, a dynamic collaboration approach is designed. Different from traditional alternative iteration algorithms, our proposed dynamic collaboration approach triggers the updates of the sub-swarm division and UAV trajectory, and periodically updates the transmission power. Simulation results validate that the proposed dynamic collaboration approach reduces average total interference, tracking steps, and target switching times by $28\%$, $33\%$, and $48\%$, respectively, comparing to existing baselines.
\end{abstract}

\begin{IEEEkeywords}
Multi-target tracking, target association, artificial potential field, mean field game, dynamic collaboration.
\end{IEEEkeywords}

%
\IEEEpeerreviewmaketitle

\section{Introduction}\label{section1}
%
%
%
%

\IEEEPARstart{T}{he} past decades have witnessed the flourish of unmanned aerial vehicles (UAVs) in many industries, since UAVs can provide both ubiquitous tracking and communication functions thanks to their small size, easy deployment, and high flexibility~\cite{07}. The drawbacks of a single UAV are limited energy, short range of communication, low bandwidth, and handover\cite{08}, which makes it unable to complete complex tasks. To perform various complex tasks, UAVs can constitute swarms by cooperating with each other\cite{46}. For example, for traffic monitoring, event tracking, military reconnaissance, or missile tracking, a swarm of UAVs can not only track multiple targets simultaneously, but also communicate with each other to share their acquired information of targets.

Most of existing works aim to enhance the performance of UAV swarm or multi-UAV networks from the perspective of tracking or communications, separately. As for improving multi-target tracking (MTT) performance\cite{34,35,36,45}, a cooperative tracking architecture is designed for UAV swarm to save energy consumption, reduce execution latency, and increase trajectory prediction accuracy in~\cite{35}. A decentralized control method is proposed to guide UAV swarms for single target tracking and MTT in \cite{45}. Note that the communication between UAVs to exchange information of targets is highly susceptible to all kinds of interferences, and thus affects the tracking quality. However, most works ignore the impact of communication on MTT performance.

A UAV suffers from severe inter-UAV interference when UAVs communicate with others in the swarm during MTT. Besides, the communication quality between UAVs is particularly vulnerable to jamming\cite{29}. Recently, there are growing research efforts for interference management in multi-UAV communication networks\cite{31,32,33,06,48}. In \cite{31}, a joint control scheme of trajectory and power is developed to promote the sum data transmission rate of the UAV-enabled interference channel for a given flight interval. To support reliable communication between the UAV swarm and the ground control unit while effectively reducing the impact of adjacent channel interference and the external interference, a priority-based resource coordination scheme is proposed to enhance the minimum signal-to-interference-plus-noise ratio (SINR) among multiple UAVs in \cite{32}. The malicious jamming is considered in \cite{02}, where a joint optimization approach of UAVs' trajectories and clustering is proposed to significantly improve the sum achievable rate for the UAV swarm. 

To promote the performance of UAV swarm from both tracking and communications, our preliminary work proposes a dual-field-cooperation approach to jointly adjust trajectory and power for UAV swarm during target tracking\cite{06}, which only considers a single target tracking scenario and inter-UAV interference. To the best of our knowledge, the performance optimization of tracking and communication for a multi-UAV MTT system with jamming is unexplored in the existing literature. To this end, this paper studies an interference minimization problem in a MTT scenario where a UAV swarm tracks multiple targets under dynamic jamming.

Compared to \cite{06}, the introduction of multi-target and dynamic jamming makes the interference problem more complex and challenging. The critical challenges to mitigating the inter-UAV interference and jamming for MTT are analyzed as follows. {\em First}, during MTT, the UAV's associated target, trajectory, and power can affect the inter-UAV interference and jamming, and thus influence the quality of MTT. Specifically, driven by a MTT mission, UAVs need to associate with the target before tracking. Different target association results, containing different sub-swarm division and match results, lead to different trajectories of UAVs. Meanwhile, the trajectory determines the distance among UAVs and the distance between UAVs and jammers, changing the inter-UAV interference and jamming\cite{30}. Moreover, power is an important element in affecting interference. In such a large-scale scenario with multiple UAVs, it is difficult to propose an effective distributed power control policy to reduce the interference\cite{20}. {\em Second}, the UAVs' associated target, trajectory, and power are not independent of each other. Specifically, UAV's associated target directly determines its trajectory, and the trajectory affects the power control through path loss, which makes interference mitigation more complex. {\em Third}, in this highly dynamic scenario, the mobility of the UAVs, jammers, and targets makes inter-UAV interference and jamming time-varying. Therefore, it is necessary to propose a reasonable dynamic adjustment approach to target association, trajectory, and power to adapt to the change of inter-UAV interference and jamming.

To tackle above challenges, we propose a dynamic collaboration approach of target association, trajectory planning, and power control. The contributions of this work are summarized below.
\begin{itemize}
  \item We study a novel scenario where a UAV swarm tracks multiple targets under dynamic malicious jamming. An interference minimization problem is formulated to minimize the total interference, including inter-UAV interference and jamming, under the constraints of UAVs' flying abilities, battery capacity, and communication quality.

  \item We propose a dynamic collaboration algorithm to solve the interference minimization problem. The optimization problem is characterized by multiple coupling and difficult to solve directly. By exploiting the characteristics of the problem, we decompose it into three subproblems: target association, trajectory planning, and power control. First, we propose a cluster-evolutionary target association (CETA) algorithm to divide UAV swarm into multiple sub-swarms and match them to targets. Second, a jamming-sensitive and singular case tolerance (JSSCT)-artificial potential field (APF) algorithm is constructed to plan UAVs' trajectories and avoid the collision with obstacles, other UAVs, and jammers (in the environment). Third, the subproblem of power control is reformulated as a jamming-aware mean field game (JA-MFG), which considers the inter-UAV interference and jamming simultaneously. Finally, to minimize the total interference, a dynamic collaboration approach is designed to adaptively update the sub-swarms division, UAV trajectory, and transmission power and adapt to the changes in the UAVs', jammers', and targets' locations.

  \item The performance of the proposed dynamic collaboration approach is validated by simulations. Numerical results demonstrate that the proposed CETA algorithm can effectively mitigate the interference compared with the other popular clustering algorithms. Moreover, the proposed dynamic collaboration approach considering jammer reduces the interference significantly and improves the tracking quality greatly compared with the state-of-art baselines.
\end{itemize}

The rest of this paper is outlined as follows. The system model of considered MTT scenario is introduced, and then the interference minimization problem is formulated in Sec.~\ref{section2}. In Sec.~\ref{section3}, we propose a dynamic collaboration approach for target association, trajectory planning, and power control. Section.~\ref{section4} shows the numerical results. Finally, we conclude the paper in Sec.~\ref{section5}. Table~\ref{parameters and definition} lists notation in this paper.
\begin{table*}
\caption{Notation}
\label{parameters and definition}
\centering
\scalebox{0.9}{
\begin{tabular}{|c|c|c|c|}
\hline
Parameters & Description & Parameters & Description\\
\hline
$N, \mathcal{N}$ & Number of UAVs and set of UAVs & $\Lambda$ & Number of slots\\
\hline
$K$ & Number of obstacles & $v_{i}$ & Speed of the $i$-th UAV\\
\hline
$M$ & Number of targets and malicious jammers & $\alpha$ & Path loss\\
\hline
$\Phi_{m}$ &  Number of UAVs in the $m$-th sub-swarm & $\gamma_{\rm th}$ & Minimum SINR requirement\\
\hline
$p_i^{\rm U}, p_m^{\rm J}$ & Transmit power of the $i$-th UAV and jamming power of the $m$-th jammer & $d_c$ & Cutoff distance\\
\hline
$g_{n,i^{'}}^{\rm U}, g_{m,i^{'}}^{\rm J}$ &{\makecell[c]{Interference channel gain between the $i^{'}$-th UAV and the $n$-th UAV and \\jamming channel gain between the $i^{'}$-th UAV and the $m$-th jammer}} & $r_{m}$ & Centroid of the $m$-th sub-swarm\\
\hline
$X_{i}, X_{k}, X_{m}, X_{\rm tar}^{m}$ & 3-D position of the $i$-th UAV, the $k$-th obstacle, the $m$-th jammer, and the $m$-th target, respectively & $d_0$ & Influence distance\\
\hline
$e(0), e_{i}$ & Initial energy and remaining energy of the $i$-th UAV & $\widetilde{L}$ & Step length\\
\hline
$\delta_{i}$ & Minimum distance between the $i$-th UAV and other UAVs with higher density & $\rho_{i}$ & Local density of the $i$-th UAV\\
\hline
$k_{\rm rep}$, $k_{\rm obs}$, $k_{\rm jam}$ & {\makecell[c]{Repulsive gain coefficient between UAVs, \\between UAVs and obstacles, and between UAVs and jammers, respectively}} & $k_{\rm att}$ & Attractive gain coefficient\\
\hline
$I_{\rm th}^{\rm U}$ , $I_{\rm th}^{\rm J}$ & Interference threshold and jamming threshold & $k_{\rm ext}$ & External force gain coefficient\\
\hline
\end{tabular}
}
\end{table*}
\section{System Model and Problem Formulation}\label{section2}
In this section, we first establish the model of MTT using a swarm of UAVs under malicious jammers in Sec.~\ref{section2A}. Next, in Sec.~\ref{section2B}, we formulate the joint target association, path, and power optimization problem with the target of minimizing interference.
\subsection{System Model}\label{section2A}
We consider a scenario in which a UAV swarm tracks $M$ targets under $M$ malicious jammers while avoiding collision with $K$ obstacles, as shown in Fig. \ref{figure1a}. The swarm contains $N$ UAVs, and these $N$ UAVs are divided into $M$ sub-swarms. It takes $\Lambda T$ seconds for the swarm to track $M$ targets, which moves with the same uniform speed in a straight line towards different directions. Each target is accompanied by a malicious jammer, which moves in the same way as its following target and keeps a certain distance from it. Each UAV communicates with the nearest UAV to convey the collected information in the same channel. However, it does not forward the information obtained from other UAVs in the process of tracking. The malicious jammers generate strong jamming to degrade the communication quality between UAVs, thus preventing UAVs from tracking targets. When UAV$_{i}$ communicates with the closest UAV$_{i^{'}}$, UAV$_{i^{'}}$ receives interference brought by other UAVs and jamming sent by malicious jammers. Therefore, the SINR of UAV$_{i^{'}}$ is
\begin{equation}\label{SINR of general player}
{\rm{SINR}}_{i,i^{'}}(t)=\frac{p_i^{\rm U}(t)g_{i,i^{'}}^{\rm U}(t)}{I_{i^{'}}^{\rm U}(t)+I_{i^{'}}^{\rm J}(t)+\sigma(t)^2},
\end{equation}
where $p_i^{\rm U}(t)$ is the transmit power of UAV$_{i}$. $g_{i,i^{'}}^{\rm U}(t)=h_{i,i^{'}}^{\rm U}(t)d(X_{i}(t),X_{i^{'}}(t))^{-\alpha}$ denotes the channel gain from UAV$_{i}$ to UAV$_{i^{'}}$, where $h_{i,i^{'}}^{\rm U}(t)$ follows the Nakagami-m distribution\cite{01}, $d(\cdot)$ is the distance function, $X(t)=(x(t),y(t),z(t))$ is the position coordinates in $3$-D space, and $\alpha$ denotes the path loss exponent. $\sigma(t)^2$ denotes the background noise power. $\textstyle{I_{i^{'}}^{\rm U}(t)=\sum\limits_{n\neq{i,i^{'}}}^{N}p_n^{\rm U}(t)g_{n,i^{'}}^{\rm U}(t)}$ and $\textstyle{I_{i^{'}}^{\rm J}(t)=\sum\limits_{m=1}^{M}p_m^{\rm J}(t)g_{m,i^{'}}^{\rm J}(t)}$ denote the inter-UAV interference from other UAVs to UAV$_{i^{'}}$ and the jamming received by UAV$_{i^{'}}$ from jammers, respectively,
where $p_m^{\rm J}(t)$ is the jamming power of the $m$-\rm{th} jammer, and $g_{n,i^{'}}^{\rm U}(t)$ and $g_{m,i^{'}}^{\rm J}(t)$ denote the interference channel gain and the jamming channel gain, respectively.
\subsection{Problem Formulation}\label{section2B}
 An optimization problem is formulated for minimizing the total interference of each UAV subject to UAV's flying abilities, battery capacity, and communication quality, which is given as
\begin{align}
\min_{X_{i}(t),\Delta\beta_{i}(t),\Phi_{m}(t),p_{i}^{\rm U}(t)}&I_{i^{'}}^{\rm U}(t)+I_{i^{'}}^{\rm J}(t)\tag{P1}\label{P1}\\
{\rm{s.t.}} &\sum\limits_{m=1}^{M}\Phi_{m}(t)=N,\tag{P1a}\label{P1a}\\
     &\bigcap\limits_{i=1}^N X_{i}(t)=\emptyset,\tag{P1b}\label{P1b}\\
     &X_{i}(t)\cap\Omega_{k}=\emptyset,\tag{P1c}\label{P1c}\\
     &X_{i}(t)\cap X_{m}(t)=\emptyset,\tag{P1d}\label{P1d}\\
     &d(X_{i}(\Lambda T),X_{m}^{\rm tar}(\Lambda T))\le d_{\rm max},\tag{P1e}\label{P1e}\\
     &0^{\circ}\leq\Delta\beta_{i}(t)<90^{\circ},\tag{P1f}\label{P1f}\\
     &V_i(t)\leq V_{\rm max},\tag{P1g}\label{P1g}\\
     &e_{i}(t)\in[0,e_{i}(0)],\tag{P1h}\label{P1h}\\
     &p_i^{\rm U}(t)\in[0,p_{\rm max}^{\rm U}],\tag{P1i}\label{P1i}\\
     &p_m^{\rm J}(t)\in[0,p_{\rm max}^{\rm J}],\tag{P1j}\label{P1j}\\
     &{\rm d}e_{i}(t)\!=\!-p_{i}^{\rm U}(t){\rm d}t+\nu_{t}{\rm d}W_{i}(t),\tag{P1k}\label{P1k}\\
     &{\rm SINR}_{i,i^{'}}(t)\ge\gamma_{\rm th},\tag{P1l}\label{P1l}
\end{align}
where ${\forall t}\in[0,\Lambda T]$, ${\forall i}\in[1,N]$, ${\forall k}\in[1,K]$, and ${\forall m}\in[1,M]$.
\begin{figure}
\centering
\includegraphics[width=3.5in]{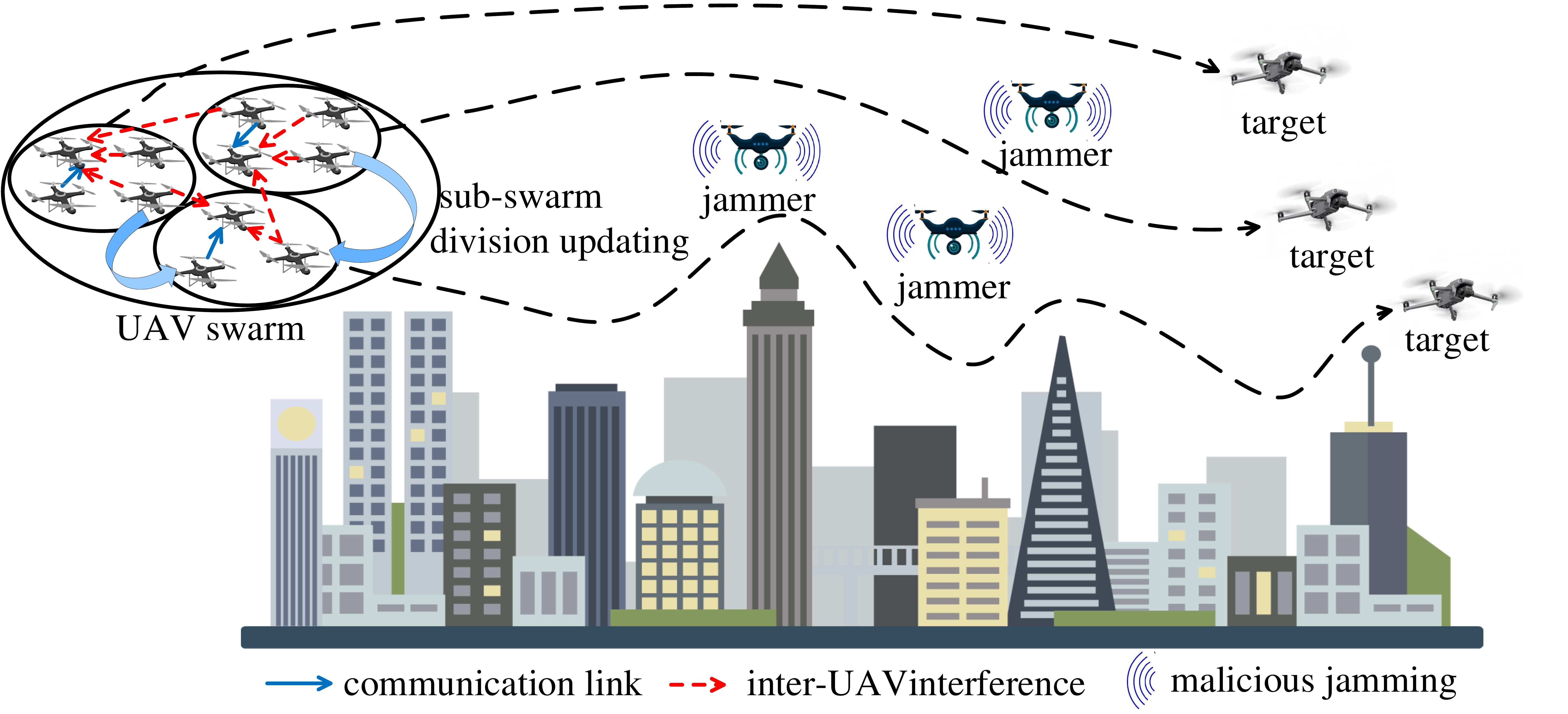}
\caption{A UAV swarm avoids the obstacles (buildings with different shapes on the ground) and tracks multiple targets under malicious jammers.}
\label{figure1a}
\end{figure}
Constraint (\ref{P1a}) ensures the sum of the UAVs in each sub-swarm is equal to $N$, where $\Phi_{m}(t)$ is the number of UAVs in $m$-\rm{th} sub-swarm. Constraint (\ref{P1b}) prevents physical collision between UAVs\footnote{Note that the different UAVs can be in the same position at the different moments.}. Constraints (\ref{P1c}) and (\ref{P1d}) ensure the collision does not occur between UAVs and obstacles and between UAVs and jammers, where $\Omega_{k}$ is the coordinate set of all points on $k$-{\rm th} obstacle surface. Constraint (\ref{P1e}) ensures that UAVs can reach the target. Constraint (\ref{P1f}) assures that a UAV does not jitter, where $\Delta\beta_{i}(t)$ is the turning angle. Constraint (\ref{P1g}) is the flight speed constraint of UAV, where $V_i(t)$ is the flight speed, $V_i(t)=\sqrt{(\dot x_{i}(t))^2+(\dot y_{i}(t))^2+(\dot z_{i}(t))^2}$ with $\dot x_{i}(t)$, $\dot y_{i}(t)$, and  $\dot z_{i}(t)$ denoting the time-derivatives of $x_{i}(t)$, $y_{i}(t)$, and $z_{i}(t)$, respectively, and $V_{\rm max}$ in meter/second (m/s) is maximum speed. Constraint (\ref{P1h}) indicates that the remaining energy does not exceed the initial energy, where $e_{i}(t)$ is the remaining energy, and $e_{i}(0)$ represents the initial energy\footnote{Note that our minimization objective interference is mainly affected by the communication power and not by the power consumed in flight in this paper. Therefore, we only consider the communication energy, and the energy for UAV flight is not considered.}. Constraints (\ref{P1i}) and (\ref{P1j}) are the transmit power constraint of UAVs and jammers, respectively, where $p_{\rm max}^{\rm U}$ and $p_{\rm max}^{\rm J}$ are the maximum transmit power of UAVs and jammers, respectively. Constraint (\ref{P1k}) shows the relationship between the remaining energy and the transmit power of UAVs. As in \cite{39,40}, the evolutionary rule of the remaining energy is modeled as an Ito process in constraint (\ref{P1k}), where $\nu_{t}$ is a constant, and $W_{i}(t)$ is an independent Wiener process. Constraint (\ref{P1l}) means that each UAV should meet the SINR requirement, where $\gamma_{\rm th}$ is the SINR threshold.

\section{Dynamic Collaboration of Target Association, Trajectory Planning, and Power Control}\label{section3}
Problem (\ref{P1}) is difficult to solve directly owing to the multiple coupling. Specifically, the UAV's associated target and trajectory are coupled tightly. Meanwhile, the transmit power relies on the UAV trajectory. To tackle the problem (\ref{P1}) effectively, it is decomposed into three subproblems: 1) target association subject to (\ref{P1a}), 2) trajectory planning subject to (\ref{P1b})-(\ref{P1g}), and 3) power control subject to (\ref{P1h})-(\ref{P1l}).

To better understand the dynamic collaboration workflow of these three subproblems, we illustrate their execution order in Fig. \ref{figure1b}. As we can see, it takes $\Lambda T$ seconds for the UAV swarm to track $M$ targets. To adaptively control UAVs, the tracking duration $\Lambda T$ is divided into $\Lambda$ slots, and each slot contains $Y$ subslots with time intervals $\Delta t$. Before UAVs perform multiple target tracking, the target association is carried out to partition swarm according to UAVs' starting positions and match each sub-swarm to a target. Then, the trajectory is planned for UAVs to track the targets. At the end of each slot, the power is allocated according to the distance among UAVs, and thus the inter-UAV interference and jamming can be updated. Note that, in our system, the target association, trajectory planning, and power control are dynamically collaborated before reaching the target. In specific, as shown in Fig. \ref{figure1b}, target association is executed only when the current tracking target is not the closest one to the centroid of the corresponding sub-swarm. Besides, the trajectory is replanned only when the inter-UAV interference and jamming exceed a certain threshold.

In the following, we elaborate the solution of problem (\ref{P1}). First, a cluster-evolutionary target association (CETA) algorithm is designed in Sec.~\ref{section3A}. Second, we propose a jamming-sensitive and singular case tolerance (JSSCT)-artificial potential field (APF) algorithm in Sec.~\ref{section3B}. Then, a jamming-aware mean field game (JA-MFG) power control scheme is developed in Sec.~\ref{section3C}. Finally, to minimize the total interference, a dynamic collaboration approach is proposed in Sec.~\ref{section3D} to update the sub-swarms division, UAV trajectory, and transmit power.
\begin{figure}
  \centering
  \includegraphics[width=3.5in]{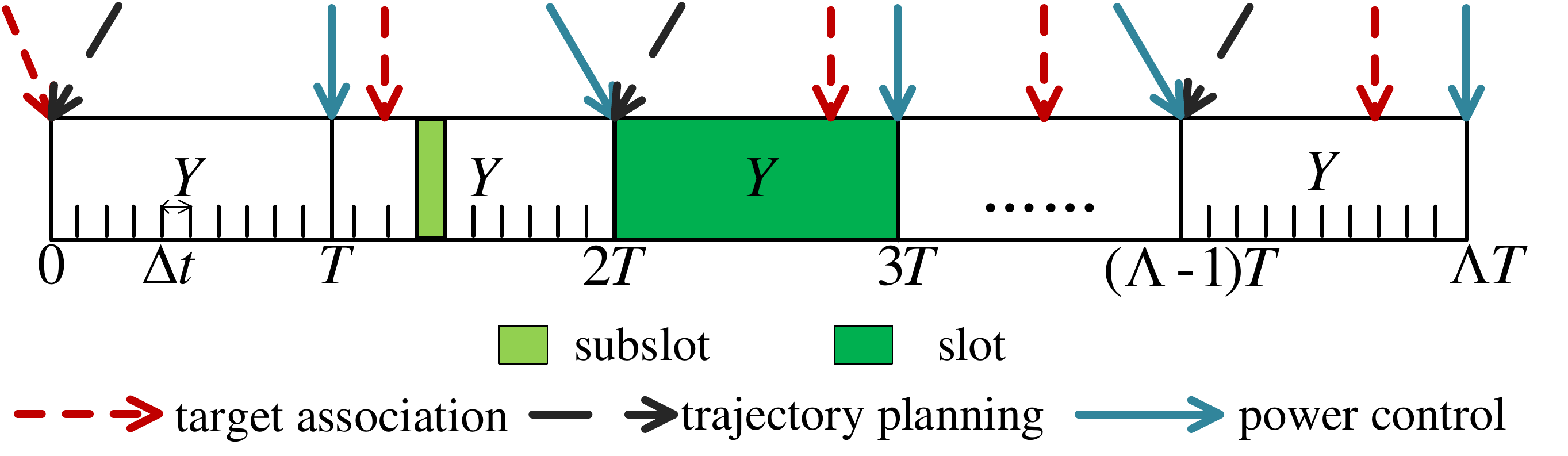}
  \caption{The execution order of target association, trajectory planning, and power control.}
  \label{figure1b}
\end{figure}
\subsection{Target Association for UAV Swarm}\label{section3A}
To perform MTT, UAVs are associated with the tracking target first. To enhance the efficiency of target association, we divide the UAV swarm into multiple sub-swarms, and then match these sub-swarms to targets. The clustering by fast search and find of density peaks (CFSFDP) algorithm has been widely used in clustering the point with arbitrary shapes\cite{12,13,14}. The superiority of this algorithm is that it requires only one-time iteration. In CFSFDP, the cluster number cannot be given directly and is found via a decision graph. However, in our scenario, the number of sub-swarms needs to be the same as the number of targets. The number of sub-swarm is obtained by using the decision graph in CFSFDP, and may not meet the aforementioned requirement. To adapt to our MTT system, inspired by the merit of CFSFDP, we propose a CETA algorithm.

The basic idea of CETA algorithm is that sub-swarm centers are encircled by UAVs with lower local densities and keep a relatively large distance from UAVs with higher local densities. For each UAV$_{i}$, the local density $\rho_{i}$ is given by
\begin{equation}\label{local density}
\rho_{i}(t)=\sum\limits_{j\neq{i}}^{N}\chi(d(X_{i}(t),X_{j}(t))-d_c),
\end{equation}
where
$\chi(d(X_{i}(t),X_{j}(t))-d_c)\!=\!\left\{
\begin{aligned}
&1,&d(X_{i}(t),X_{j}(t))<d_c,\\
&0,&d(X_{i}(t),X_{j}(t))\ge d_c,
\end{aligned}\right.$
where $d_c$ is a cutoff distance and takes the first $2\%$ of all distances empirically\cite{38}. Therefore, $\rho_{i}$ reflects the number of UAVs whose distances to UAV$_{i}$ are less than $d_c$. The distance $\delta_{i}$ from UAVs of a higher density refers to the minimum distance between UAV$_{i}$ and other UAVs with a higher density. Especially, if the UAV$_{i}$ has the highest density, its $\delta_{i}$ is equal to the maximum distances between UAV$_{i}$ and other UAVs. $\delta_{i}$ is calculated by
\begin{equation}\label{minimum distance}
\delta_{i}(t)=\left\{
\begin{aligned}
&\mathop{\rm min}_{j:\rho_{j}(t)>\rho_{i}(t)}{(d(X_{i}(t),X_{j}(t)))},&\exists{\rho_{j}(t)>\rho_{i}(t)},\\
&\mathop{\rm max}_{j\neq i,j=1,2,..n}{(d(X_{i}(t),X_{j}(t)))},&\nexists{\rho_{j}(t)>\rho_{i}(t)}.
\end{aligned}\right.
\end{equation}

The sub-swarm center is characterized by $\rho_{i}$ and $\delta_{i}$ simultaneously. To determine $M$ sub-swarm centers quickly, we define
\begin{equation}\label{cluster centers method}
\eta_{i}(t)=\rho_{i}(t)\cdot\delta_{i}(t).
\end{equation}
After calculating $\eta_{i}$ of each UAV, UAVs are sorted according to their corresponding $\eta_{i}$ in descending order, and the first $M$ UAVs are appointed as the sub-swarm center. The remaining UAVs are allotted to the sub-swarm where the nearest sub-swarm center belongs. In this way, the whole process of sub-swarm division is completed.

To achieve the tracking task, UAVs should be matched to $M$ targets. To this end, the centroid $r_{m}$ of each sub-swarm is
\begin{equation}\label{centroid of each cluster}
r_{m}(t)=\frac{\sum\limits_{i\in C_{m}(t)}X_{i}(t)}{\Phi_{m}(t)},
\end{equation}
where $C_{m}$ denotes the $m$-\rm{th} sub-swarm of UAVs. Then, we propose the principle of target association that each target is tracked by a sub-swarm of UAVs whose centroid is closest to the target. So far, the subproblem of target association is solved, and Algorithm 1 briefly illustrates the CETA algorithm.
\begin{algorithm}
\caption{CETA Algorithm}
\label{alg1}
\begin{algorithmic}[1]
\REQUIRE $X_{i}$, $N$, $M$, $d_c$, where ${\forall i}\in[1,N]$.
\FOR{$i=1$ to $N$}
\STATE Compute the local density $\rho_{i}$, the distance $\delta_{i}$, and the value of $\eta_{i}$ according to (\ref{local density}), (\ref{minimum distance}), and (\ref{cluster centers method}), respectively.
\ENDFOR
\STATE Sort UAVs according to their corresponding $\eta_{i}$ in descending order.
\STATE Select the first $M$ UAVs as the sub-swarm center.
\STATE Assign remaining UAVs to the sub-swarm where the nearest sub-swarm center belongs.
\STATE Compute the centroid of each sub-swarm according to (\ref{centroid of each cluster}).
\STATE Assign each sub-swarm of UAVs to track the target closest to the centroid.
\end{algorithmic}
\end{algorithm}
\vspace{-1.43em}
\subsection{Trajectory Planning for UAV Swarm under Malicious Jammer}\label{section3B}
The tracking target of UAVs can be obtained according to the designed target association algorithm. In this subsection, trajectory planning for the UAV swarm is performed to make UAVs avoid collision with obstacles, malicious jammers, and other UAVs in the swarm and reach the target. In specific, we propose the JSSCT-APF to plan trajectory for the UAV swarm. In JSSCT-APF, each UAV is steered towards its target by the resultant potential field of attractive and repulsive potential field. In our MTT scenario, the attractive potential field, denoted by $\varphi_{\rm att}$, is created by the target associated with each UAV to make UAVs close to their tracking targets, which can be given by
\begin{equation}\label{attractive potential field}
\varphi_{\rm att}(X_{i}(t))=\frac{1}{2}k_{\rm att}A(X_{i}(t)).
\end{equation}
where $k_{\rm att}$ is the gain factor of the attractive potential field, $A(X_{i}(t))=d(X_{i}(t),X_{m}^{\rm tar}(t))^{2}$, $X_{m}^{\rm tar}(t)$ denotes the position coordinate of the target.

The repulsive potential fields are created by other $N-1$ UAVs in the swarm, $K$ obstacles, and $M$ malicious jammers to drive UAVs away from them, which are given by
\begin{equation}\small\label{rep repulsive potential field}
\varphi_{\rm rep}(X_{i}(t))=\left\{
\begin{aligned}
&\sum\limits_{j=1,j\neq{i}}^{N}\frac{1}{2}k_{\rm rep}(\lambda)B(X_{i}(t),X_{j}(t))A(X_{i}(t))^{\frac{q}{2}},\\
&\quad\qquad\qquad\qquad\quad{d(X_{i}(t),X_{j}(t))\leq d_{0}},\\
&0,\qquad\qquad\qquad\quad{d(X_{i}(t),X_{j}(t))>d_{0}},\\
\end{aligned}\right.
\end{equation}
where $\varphi_{\rm rep}$ is repulsive potential field created by other UAVs in the swarm. $k_{\rm rep}(\lambda)$ is the gain factor of corresponding potential field, where $\lambda={1,2,...,\Lambda}$ is the time slot index. $q$ is a positive constant, and generally set $q=1$\cite{49}. $B(X_{i}(t),X_{j}(t))=\left(\frac{1}{d(X_{i}(t),X_{j}(t))}-\frac{1}{d_{0}}\right)^{2}$, where $d_{0}$ represents the maximum influence range of obstacles. The repulsive potential fields created by obstacles and malicious jammers, denoted by $\varphi_{\rm obs}$ and $\varphi_{\rm jam}$, are obtained by the same way as $\varphi_{\rm rep}$.

Calculating the negative gradient of the potential field functions, we obtain
\begin{equation}\label{attractive force}
\vec F_{\rm att}(X_{i}(t))\!\!=\!\!-\nabla_{X_{i}}(\varphi_{\rm att}(X_{i}(t)))\!=\!k_{\rm att}{A(X_{i}(t))^{\frac{1}{2}}}\vec n_{im},
\end{equation}
\begin{equation}\small
\begin{aligned}\label{rep repulsive force}
&\vec F_{\rm rep}(X_{i}(t))=-\nabla_{X_{i}}(\varphi_{\rm rep}(X_{i}(t)))\\
&=\left\{
\begin{array}{ll}
\vec F_{\rm rep1}(X_{i}(t))+\vec F_{\rm rep2}(X_{i}(t)),&{d(X_{i}(t),X_{j}(t))\leq d_{0}},\\
0,&{d(X_{i}(t),X_{j}(t))>d_{0}},\\
\end{array} \right.
\end{aligned}
\end{equation}
where
\begin{equation}\small\label{rep1 repulsive force}
\vec F_{\rm rep1}(X_{i}(t))\!\!=\!\!\sum\limits_{j\neq{i}}^{N}\!k_{\rm rep}(\lambda)B(X_{i}(t),\!X_{j}(t))^\frac{1}{2}\frac{A(X_{i}(t))^{\frac{q}{2}}}{d(X_{i}(t),\!\!X_{j}(t))^{2}}\vec n_{ji},
\end{equation}
\begin{equation}\small\label{rep2 repulsive force}
\vec F_{\rm rep2}(X_{i}(t))\!=\!\sum\limits_{j\neq{i}}^{N}\frac{q}{2}k_{\rm rep}(\lambda)B(X_{i}(t),X_{j}(t))A(X_{i}(t))^{\frac{q}{2}-1}\vec n_{im}.
\end{equation}
where $\vec F_{\rm att}(X_{i}(t))$ and $\vec F_{\rm rep}(X_{i}(t))$ are the attractive force and the repulsive forces corresponding to $\varphi_{\rm att}$ and $\varphi_{\rm rep}$, respectively. In (\ref{attractive force}), (\ref{rep1 repulsive force}), and (\ref{rep2 repulsive force}), $\vec n_{im}$ is the unit vector pointing to the target $m$ from the UAV$_{i}$ and $\vec n_{ji}$ is the unit vector pointing to the UAV$_{i}$ from the UAV$_{j}$. Similarly, the repulsive force functions $\vec F_{\rm obs}(X_{i}(t))$ and $\vec F_{\rm jam}(X_{i}(t))$ produced by obstacles and jammers are the negative gradient of their corresponding repulsive potential field function $\varphi_{\rm obs}$ and $\varphi_{\rm jam}$, respectively, which are no longer listed here.

Thus, the resultant force is equal to
\begin{equation}\label{total force}
\begin{aligned}
\vec F_{\rm tot}(X_{i}(t))&=\vec F_{\rm att}(X_{i}(t))+\vec F_{\rm rep}(X_{i}(t))\\
&+\vec F_{\rm obs}(X_{i}(t))+\vec F_{\rm jam}(X_{i}(t)).
\end{aligned}
\end{equation}
Then, the turning angle is derived as
\begin{equation}\label{turning angle}
\Delta\beta_{i}(t)=\arccos\frac{\vec F_{\rm tot}(X_{i}(t-1))\cdot\vec F_{\rm tot}(X_{i}(t))}{\mid\vec F_{\rm tot}(X_{i}(t-1))\mid\mid\vec F_{\rm tot}(X_{i}(t))\mid}.
\end{equation}

Finally, UAVs move in the direction of resultant force $\vec F_{\rm tot}(X_{i}(t))$, and its movement determined by the following equation,
\begin{equation}\label{coordinates of next step}
\left\{
\begin{array}{l}
x_{i}(t+1)=x_{i}(t)+\widetilde{L}_{i}(t)\ast\cos\theta_{ix}(t), \\
y_{i}(t+1)=y_{i}(t)+\widetilde{L}_{i}(t)\ast\cos\theta_{iy}(t), \\
z_{i}(t+1)=z_{i}(t)+\widetilde{L}_{i}(t)\ast\cos\theta_{iz}(t),
\end{array}
\right.
\end{equation}
where $\widetilde{L}_{i}(t)=V_i(t)\Delta(t)$ is step length of UAV$_{i}$ from $t$ to $t+1$. In APF, the step length of each step is generally fixed, and its settings need to meet constraint (\ref{P1g}). $\theta_{ix}(t)$, $\theta_{iy}(t)$ and $\theta_{iz}(t)$ denote the angles between resultant force $\vec F_{\rm tot}(X_{i}(t))$ and the three coordinate axes.

However, the UAV may stop at a balance point where the resultant force is zero during the tracking. To rescue the UAV from this dilemma and make it continue to track the target, we introduce an external force considering the impact of jammers, which is given as
\begin{equation}\small\label{external force}
\vec F_{\rm ext}(X_{i}(t))\!\!=\!\!\frac{1}{2}k_{\rm ext}F_{\rm ext1}(X_{i}(t))F_{\rm ext2}(X_{i}(t))F_{\rm ext3}(X_{i}(t))\!\vec n_{im},
\end{equation}
where $F_{\rm ext1}(X_{i}(t))\!=\!\sum\limits_{j\neq{i}}^{N}B(X_{i}(t),X_{j}(t))$, $F_{\rm ext2}(X_{i}(t))\!\!=\!\!\sum\limits_{k=1}^{K}\!B(X_{i}(t),X_{k})$, $F_{\rm ext3}(X_{i}(t))\!\!=\!\!\sum\limits_{m=1}^{M}\!B(X_{i}(t),X_{m}(t))$, $k_{\rm ext}$ is gain factor.

Moreover, the UAV may jitter when the turning angle $\Delta\beta_{i}(t)$ satisfies $90^{\circ}\leq\Delta\beta_{i}(t)\leq180^{\circ}$. To avoid this phenomenon, the turning angle $\Delta\beta_{i}(t)$ is reduced by half. In this case, the principle of trajectory update is
\begin{equation}\small\label{update coordinates of next step}
\left\{
\begin{array}{l}
x_{i}(t+1)=x_{i}(t)+\widetilde{L}_{i}(t)\ast\cos(\theta_{ix}(t-1)-\frac{1}{2}\Delta\beta_{i}(t)), \\
y_{i}(t+1)=y_{i}(t)+\widetilde{L}_{i}(t)\ast\cos(\theta_{iy}(t-1)-\frac{1}{2}\Delta\beta_{i}(t)), \\
z_{i}(t+1)=z_{i}(t)+\widetilde{L}_{i}(t)\ast\cos(\theta_{iz}(t-1)+\frac{1}{2}\Delta\beta_{i}(t)).
\end{array}
\right.
\end{equation}
To this end, the subproblem of trajectory planning for UAV swarm under malicious jammer is solved, and the JSSCT-APF is shown in Algorithm 2.
\begin{algorithm}
\caption{JSSCT-APF Algorithm}
\label{alg2}
\begin{algorithmic}[2]
\REQUIRE ~~\\
Initialize $N$, $K$, $M$, $X_{i}(0)$, $X_{k}$, $X_{m}(0)$, $X_{m}^{\rm tar}(0)$, $k_{\rm att}$, $k_{\rm rep}(0)$, $k_{\rm obs}$, $k_{\rm jam}(0)$, $k_{\rm ext}$, $\widetilde{L}$, $\Lambda T$, and $d_{0}$, where ${\forall i}\in[1,N]$, ${\forall k}\in[1,K]$, and ${\forall m}\in[1,M]$.
\ENSURE $X_{i}(t)$, where ${\forall t}\in[0,\Lambda T]$, ${\forall i}\in[1,N]$.
\FOR{$t=1$ to $\Lambda T$}
\FOR{$i=1$ to $N$}
\IF{$d(X_{i}(t),X_{m}^{\rm tar}(t))>d_{\rm max}$}
\STATE Calculate the attractive force, the repulsive forces, the resultant force, and the turning angle according to (\ref{attractive force}), (\ref{rep repulsive force}), (\ref{total force}), and (\ref{turning angle}), respectively.
\IF{$F_{\rm tot}=0$}
\STATE Introduce an external force according to (\ref{external force}).
\STATE Set $F_{\rm tot}=F_{\rm ext}$ and update the trajectory according to (\ref{coordinates of next step}).
\ENDIF
\IF{$90^{\circ}\leq\Delta\beta_{i}(t)\leq180^{\circ}$}
\STATE Update the trajectory according to (\ref{update coordinates of next step}).
\ENDIF
\STATE Update the trajectory according to (\ref{coordinates of next step}).
\ENDIF
\ENDFOR
\ENDFOR
\end{algorithmic}
\end{algorithm}
\subsection{Power Control for UAV Swarm under Malicious Jamming}\label{section3C}
After the planning trajectory, the distance among UAVs can be obtained. In this subsection, we solve the subproblem of power control based on these distances to mitigate the interference.

In the subproblem of power control, UAVs are interacting with each other, which makes it difficult to solve this problem. Game theory is an effective tool to model the interactions among distributed decision makers and promote the performance of decentralized networks\cite{41}. On the one hand, the goal of power control is mitigating the total interference, containing inter-UAV interference and jamming. On the other hand, the communication quality of each UAV should meet constraint (\ref{P1l}). Substituting (\ref{SINR of general player}) into (\ref{P1l}), we can obtain
\begin{equation}\label{quality of communication}
p_i^{\rm U}(t)g_{i,i^{'}}^{\rm U}(t)-\gamma_{\rm th}(I_{i^{'}}^{\rm U}(t)+I_{i^{'}}^{\rm J}(t)+\sigma^2(t))\ge0.
\end{equation}
According to (\ref{SINR of general player}), UAVs need to increase the transmit power to improve the SINR, which makes other UAVs suffer greater interference. To balance the total interference and SINR, similar to \cite{42,44}, the constraint (\ref{P1l}) is relaxed to ${\rm SINR}_{i,i^{'}}(t)$ as close to $\gamma_{\rm th}$ as possible. Let $c_{1,i}(t)=\omega_1[p_i^{\rm U}(t)g_{i,i^{'}}^{\rm U}(t)-\gamma_{\rm th}(I_{i^{'}}^{\rm U}(t)+I_{i^{'}}^{\rm J}(t)+\sigma^2(t))]^{2}$ be the SINR cost, and $c_{2,i}(t)=\omega_2(I_{i^{'}}^{\rm U}(t)+I_{i^{'}}^{\rm J}(t))^{2}$ be the interference cost, and the cost function considering the above two costs can be expressed as
\begin{equation}\label{cost function}
c_i(t)=\omega_1c_{1,i}(t)+\omega_2c_{2,i}(t),
\end{equation}
where $\omega_1$ and $\omega_2$ are the weighting coefficients that make the above two terms fit into one scale.
The power control subproblem is equivalent to finding an optimal control policy for minimizing the cost function  over each slot subject to (\ref{P1h})-(\ref{P1k})\cite{43}. The optimal control policy is expressed as
\begin{equation}\label{optimal control policy}
Q_{i}^{\star}(t)=\mathop{\arg\min}_{p_{i}^{\rm U}(t)}{E\left[\int_{(\lambda-1)T}^{\lambda T}c_{i}(t)dt+c_{i}(\lambda T)\right]},
\end{equation}
where $c_{i}(\lambda T)$ is the terminal cost on each slot.

\subsubsection{Stochastic Differential Game Formulation}\label{section3Ca}
To obtain the above optimal control policy, a stochastic differential game ($N$-body game) is formulated to handle the power control subproblem, which is defined as follows a 5-tuple: $G=\{\mathcal{N},\{p_{i}^{\rm U}\}_{i\in\mathcal{N}},\{e_{i}\}_{i\in\mathcal{N}},\{Q_{i}\}_{i\in\mathcal{N}},\{c_{i}\}_{i\in\mathcal{N}}\}$, where $\mathcal{N}$ is the set of agents, and each UAV can be selected as an agent in the MTT system. $\{p_{i}^{\rm U}\}_{i\in\mathcal{N}}$ is actions set, which contains all possible transmit power. $\{e_{i}\}_{i\in\mathcal{N}}$ is energy state set. Constraint (\ref{P1j}) shows that there is less energy available as the agent consumes more energy. $\{Q_{i}\}_{i\in\mathcal{N}}$ is control policy and defined in (\ref{optimal control policy}). $\{c_{i}\}_{i\in\mathcal{N}}$ is cost function, and defined in (\ref{cost function}), which considers both total interference and SINR performance.

For the formulated game $G$, the value function is defined as
\begin{equation}\label{value function}
u_{i}(t,e_{i}(t))=\mathop{\min}_{p_{i}^{\rm U}(t)}{E\left[\int_{t}^{\lambda T}c_{i}(t)dt+u_{i}(T,e_{i}(\lambda T))\right]},
\end{equation}
where $t\in[(\lambda-1)T,\lambda T]$, $u_{i}(T,e_{i}(\lambda T))$ is final value of $u_{i}(t,e_{i}(t))$. According to the optimal control theory, (\ref{value function}) should meet the following Hamilton-Jacobi-Bellman (HJB) partial differential equation (PDE)\cite{23},
\begin{equation}\label{HJB}
\frac{\partial{u_{i}(t,e_{i}(t))}}{\partial{t}}+H\left(e_{i}(t),\frac{\partial{u_{i}(t,e)}}{\partial{e}}\right)
+\frac{\nu_{t}^{2}}{2}\frac{\partial^{2}{u_{i}(t,e_{i}(t))}}{\partial^{2}{e}}=0,
\end{equation}
where
\begin{equation}\label{Hamiltionian}
H\!\left(e_{i}(t),\frac{\partial{u_{i}(t,e)}}{\partial{e}}\right)
\!=\!\mathop{\min}_{p_{i}^{\rm U}(t)}\!\!\left(c_{i}(p_{i}^{\rm U}(t))-p_{i}^{\rm U}(t)\frac{\partial{u_{i}(t,e_{i}(t))}}{\partial{e}}\right)
\end{equation}
is called the Hamiltionian.

In our formulated power control stochastic differential game, $N$ coupled HJB equations needed to be solved to obtain equilibrium. However, when $N$ is sufficiently large, it is impractical to get the optimal control policy of each UAV. To deal with this problem, we introduce MFG, which can be regarded as an extension of the stochastic differential game, to our model and propose a JA-MFG power control scheme.
\subsubsection{Mean Field Approximation}\label{section3Cb}
According to (\ref{cost function}) and (\ref{optimal control policy}), we need to calculate the inter-UAV interference and jamming perceived by agent $i$ when minimizing the cost function of each slot. It requires extensive messages exchange among the $N$ agents. Therefore, before establishing the JA-MFG, we adopt the mean field approximation to the inter-UAV interference $I_{i^{'}}^{\rm U}(t)$ and jamming $I_{i^{'}}^{\rm J}(t)$. The basic idea of the mean field approximation is that all other UAVs and jammers use the same average channel gain $\overline{g}^{\rm U}$ and $\overline{g}^{\rm J}$ to UAV$_{i^{'}}$. Based on this idea, $I_{i^{'}}^{\rm U}(t)$ and $I_{i^{'}}^{\rm J}(t)$ are simplified with the inter-UAV interference mean field and jamming mean field.
In our JA-MFG power control scheme, the mean field is defined as the probability density distribution of the energy states of all UAVs at time $t$, and it can be expressed as
\begin{equation}\label{mean field}
m(t,e)=\lim\limits_{N\to\infty}\frac{1}{N}\sum\limits_{i=1}^{N}\mathds{1}_{e_{i}(t)=e},
\end{equation}
where
$\mathds{1}_{e_{i}(t)=e}=\left\{
\begin{aligned}
&1,e_{i}(t)=e,\\
&0,e_{i}(t)\neq e.
\end{aligned}\right.$

\begin{proposition}
Assume that the transmit power of UAVs and jammers are $p_{(\cdot)}^{\rm U}$ and $p_{(\cdot)}^{\rm J}$, which are predefined. Let $p_{(\cdot)}^{\rm R}$ be the received power. When the channel gain between UAVs and between UAVs and jammers are $\overline{g}^{\rm U}$ and $\overline{g}^{\rm J}$ given above, the interference mean field and jamming mean field are given as,
\begin{equation}\small\label{detail interference mean field}
\begin{aligned}
\overline{I}_{i^{'}}^{\rm U}(t)&\approx (N-2){p}^{\rm U}(t)\overline{g}^{\rm U}(t)\\ &={p}^{\rm U}(t)\left(\frac{p_{1}^{\rm R}(t)p_{2}^{\rm J}(t)-p_{2}^{\rm R}(t)p_{1}^{\rm J}(t)}{p_{1}^{\rm U}(t)p_{2}^{\rm J}(t)-p_{2}^{\rm U}(t)p_{1}^{\rm J}(t)}-g_{i,i^{'}}^{\rm U}(t)\right),
\end{aligned}
\end{equation}

\begin{equation}\small\label{detail jamming mean field}
\overline{I}_{i^{'}}^{\rm J}(t)\!\approx\! Mp_m^{\rm J}(t)\overline{g}^{\rm J}(t)\!=\!p_m^{\rm J}(t)\frac{p_{1}^{\rm R}(t)p_{2}^{\rm U}(t)-p_{2}^{\rm R}(t)p_{1}^{\rm U}(t)}{p_{1}^{\rm J}(t)p_{2}^{\rm U}(t)-p_{2}^{\rm J}(t)p_{1}^{\rm U}(t)}.
\end{equation}
\label{thm-1}
\end{proposition}

\begin{proof}\renewcommand{\qedsymbol}{}
The detailed proof is given in Appendix A.
\end{proof}
According to Proposition 1, the cost function for UAV$_{i^{'}}$ in (\ref{cost function}) can be reexpressed as
\begin{equation}\label{mean field cost function}
\begin{aligned}
c_i(t)&=\omega_1[p_i^{\rm U}(t)g_{i,i^{'}}^{\rm U}(t)-\gamma_{\rm th}(\overline{I}_{i^{'}}^{\rm U}(t)+\overline{I}_{i^{'}}^{\rm J}(t)+\sigma(t)^2)]^{2}\\
&+\omega_2(\overline{I}_{i^{'}}^{\rm U}(t)+\overline{I}_{i^{'}}^{\rm J}(t))^{2}.
\end{aligned}
\end{equation}
\subsubsection{Jamming-Aware Mean Field Game}\label{section3Cc}
In the JA-MFG, each agent is identical and exchangeable and adopts the same control policy. The agent only cares about collective policies of all other agents rather than each others' states\cite{24}. This collective behavior is reflected in (\ref{mean field}). Hence, JA-MFG can be regarded as a $2$-body game. One participant is an agent, determining the optimal control policy, and the evolution of the optimal control policy is described by the HJB backward equation. The other one is a mean field, and its evolution is described by the Fokker-Planck-Kolmogorov (FPK) forward PDE~\cite{25}.
\begin{proposition}
Considering the HJB backward equation describes the evolution of optimal control of agents, thus the optimal power is expressed as
\begin{equation}\label{optimal power}
\widetilde p^{\rm U}(t,e)=\frac{\gamma_{\rm th}(\overline{I}_{i^{'}}^{\rm U}(t)+\overline{I}_{i^{'}}^{\rm J}(t)+\sigma(t)^2)}{g_{i,j}^{\rm U}(t)}+\frac{\frac{\partial{u(t,e(t))}}{\partial{e}}}{2\omega_1{g_{i,j}^{\rm U}}(t)^{2}}.
\end{equation}
\label{thm-2}
\end{proposition}

\begin{proof}\renewcommand{\qedsymbol}{}
Please see Appendix B.
\end{proof}

According to the above definition of the relationship between the mean field and the FPK equation, we can get the FPK equation as follows
\begin{equation}\label{FPK}
\frac{\partial{m(t,e)}}{\partial{t}}+\frac{\partial}{\partial{e}}(m(t,e)p^{\rm U}(t,e))+\frac{\nu_{t}^{2}}{2}\frac{\partial^{2}{m(t,e)}}{\partial^{2}{e}}=0.
\end{equation}

To sum up, the optimal control for each agent is dominated by the HJB equation and relies on the mean field of agents. Meanwhile, the evolutionary rule of the mean field is described by the FPK equation and affected by optimal control policy. Therefore, the JA-MFG can be characterized by the two coupled PDEs, one HJB equation, and one FPK equation. In this way, solving $N$ PDEs is transited into solving these two PDEs.

\subsubsection{Solution to Jamming-Aware Mean Field Game}\label{section3Cd}
To make the above two PDEs tractable, the finite difference method is utilized to replace them with difference equations\cite{26}. First, we discretize the slot $[(\lambda-1)T,\lambda T]$ into $Y$ time points, and the time interval is $\Delta t=\frac{T}{Y}$, as discussed in Sec.~\ref{section3}. The initial energy $e(0)$ is discretized into $Z$ energy points, and the energy interval is $\Delta e=\frac{e(0)}{Z}$.

Then, we introduce the finite difference operators as follows,
\begin{equation}\label{first derivatives with respect to time}
\frac{\partial{\Upsilon(t,e)}}{\partial{t}}=\frac{\Upsilon(t+1,e)-\Upsilon(t,e)}{\Delta t},
\end{equation}
\begin{equation}\label{first derivatives with respect to state}
\frac{\partial{\Upsilon(t,e)}}{\partial{e}}=\frac{\Upsilon(t,e)-\Upsilon(t,e-1)}{\Delta e},
\end{equation}
\begin{equation}\label{second derivatives with respect to time}
\frac{\partial^{2}{\Upsilon(t,e)}}{\partial^{2}{e}}=\frac{\Upsilon(t,e+1)-2\Upsilon(t,e)+\Upsilon(t,e-1)}{\Delta e}.
\end{equation}
Substituting (\ref{first derivatives with respect to time}), (\ref{first derivatives with respect to state}), and (\ref{second derivatives with respect to time}) into (\ref{FPK}), the FPK equation can be discretized as (\ref{FPK equation}), which can be rewritten as
\begin{figure*}
\hrulefill
\begin{equation}\label{FPK equation}
\frac{m(t+1,e)-m(t,e)}{\Delta t}+\frac{m(t,e)p^{\rm U}(t,e)-m(t,e-1)p^{\rm U}(t,e-1)}{\Delta e}+\frac{\nu_{t}^{2}}{2}\frac{m(t,e+1)-2m(t,e)+m(t,e-1)}{\Delta e}^{2}=0.
\end{equation}
\end{figure*}
\begin{equation}\label{FPK equation variation}
\begin{aligned}
m(t+1,e)&=m(t,e)\\
&+\frac{\Delta t}{\Delta e}[m(t,e-1)p^{\rm U}(t,e-1)-m(t,e)p^{\rm U}(t,e)]\\
&+\frac{\nu_{t}^{2}\Delta t}{2(\Delta e)^{2}}[2m(t,e)-m(t,e+1)-m(t,e-1)].
\end{aligned}
\end{equation}
Similarly, we can solve the HJB equation, which is shown in (\ref{HJB evolution equation}).
\begin{figure*}
\begin{equation}\label{HJB evolution equation}
u(t-1,e)=u(t,e)+\Delta t[c(p^{\rm U}(t,e)m(t,e))-p^{\rm U}(t,e)\frac{u(t,e)-u(t,e-1)}{\Delta e}]+\frac{\nu_{t}^{2}\Delta t}{2(\Delta e)^{2}}[u(t,e+1)-2u(t,e)+u(t,e-1)].
\end{equation}
\hrulefill
\end{figure*}

Since the optimal control and the mean field interact with each other, we need to update (\ref{optimal power}), (\ref{FPK equation variation}), and (\ref{HJB evolution equation}) alternately till the game reaches mean field equilibrium (MFE), which is described in Algorithm 3. So far, the solution to the subproblem of power control is achieved in all slots.
\begin{algorithm}
\caption{JA-MFG Power Control Algorithm}
\label{alg3}
\begin{algorithmic}[3]
\REQUIRE ~~\\
Initialize $p^{\rm U}(t,e), m((\lambda-1)T,e), u(\lambda T,e)=0, k=1$.\\
Set $I_{\rm max}$ as the maximum iterations.\\
\WHILE {$k<I_{\rm max}$}
\STATE Update $m(t+1,e)$, $u(t-1,e)$, and $\widetilde p^{\rm U}(t,e)$ using (\ref{FPK equation variation}), (\ref{HJB evolution equation}), and (\ref{optimal power}), respectively.
\STATE $k=k+1$.
\ENDWHILE
\end{algorithmic}
\end{algorithm}

\subsection{Dynamic Collaboration Approach}\label{section3D}
The subproblems of target association, trajectory planning, and power control are addressed in the above three subsections, respectively. However, the minimum total interference in problem (\ref{P1}) cannot be obtained by solving these subproblems separately. To this end, we propose a dynamic collaboration approach to adaptively update the sub-swarms division, UAV trajectory, and transmit power. The detailed collaboration procedure is shown in
Fig. \ref{fig2}.

\begin{figure}
  \centering
  \includegraphics[width=5cm]{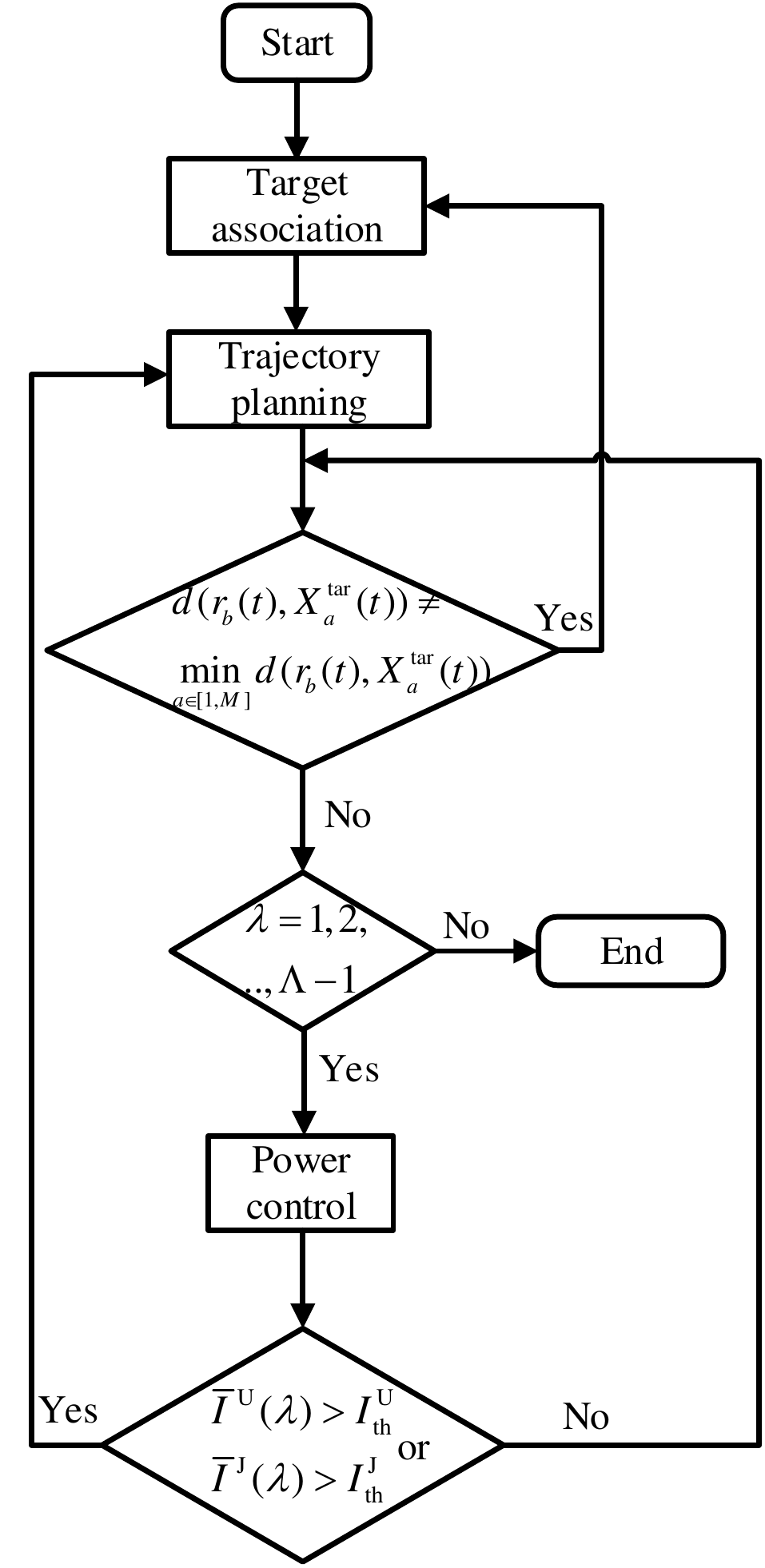}
  \caption{Dynamic collaboration approach of target association, trajectory planning, and power control.}
  \label{fig2}
\end{figure}

First, the CETA algorithm is adopted to divide UAV swarm into multiple sub-swarms and match them to targets, which determines the trajectory of UAV by the associated tracking target. At the same time, the JSSCT-APF is invoked to plan the trajectory for UAV. According to the trajectory of UAV, the distance between UAVs and their associated targets can be obtained. Once this distance changes greatly, we need to change the tracking target of the UAV immediately. This motivates us to reassociate target, where the UAV swarm is redivided into sub-swarms and associate with the corresponding targets as shown in Sec.~\ref{section3A}. To reduce the complexity of frequent target association, we propose a criterion for target reassociation. Specifically, it is supposed that the $a$-th target is tracked by UAV$_{i}$ currently, and UAV$_{i}$ belongs to $b$-th sub-swarm. When the $a$-th target is not the closest one to the centroid of the $b$-th sub-swarm, i.e., ${\forall b\in[1,M]},{d(r_{b}(t),X_{a}^{\rm tar}(t))\neq\mathop{\min}\limits_{a\in[1,M]}d(r_{b}(t),X_{a}^{\rm tar}(t))}$, the target association is triggered.

Second, the trajectory influences power control by the distance among UAVs and that between UAVs and jammers. Meanwhile, the JA-MFG power control algorithm is employed, which has an impact on the total interference. We intend to make power have an effect on trajectory through total interference. Therefore, we propose the adjustment of the repulsive force gain coefficient $k_{\rm rep}(\lambda)$ and $k_{\rm jam}(\lambda)$ according to inter-UAV interference and jamming, respectively. Specifically, $\overline{I}^{\rm U}(\lambda)$ and $\overline{I}^{\rm J}(\lambda)$ are calculated according to (\ref{average interference}). When they surpass their respective threshold ${I_{\rm th}^{\rm U}}$ and ${I_{\rm th}^{\rm J}}$, $k_{\rm rep}(\lambda)$ and $k_{\rm jam}(\lambda)$ are adjusted according to (\ref{rep variation}) and (\ref{jam variation}), and the trajectory planning is triggered.

Given the above, the proposed dynamic collaboration approach minimizes the total interference by jointly control target association, trajectory, and power. The dynamic collaboration approach of target association, trajectory planning, and power control is shown in Algorithm 4.
\begin{equation}\label{average interference}
\overline{I}^{\rm U}(\lambda)=\frac{1}{Y}\sum\limits_{t=1}^{Y}\overline{I}_{i^{'}}^{\rm U}(t),\quad\overline{I}^{\rm J}(\lambda)=\frac{1}{Y}\sum\limits_{t=1}^{Y}\overline{I}_{i^{'}}^{\rm J}(t).
\end{equation}
\begin{equation}\label{rep variation}
k_{\rm rep}(\lambda+1)=\left\{
\begin{array}{lr}
k_{\rm rep}(\lambda)+\frac{\overline{I}^{\rm U}(\lambda)-I_{\rm th}^{\rm U}}{I_{\rm th}^{\rm U}},{\overline{I}^{\rm U}(\lambda)>{I_{\rm th}^{\rm U}}},\\
k_{\rm rep}(\lambda),\qquad\qquad\quad{\overline{I}^{\rm U}(\lambda)\leq{I_{\rm th}^{\rm U}}}.\\
\end{array} \right.
\end{equation}
\begin{equation}\label{jam variation}
k_{\rm jam}(\lambda+1)=\left\{
\begin{array}{lr}
k_{\rm jam}(\lambda)+\frac{\overline{I}^{\rm J}(\lambda)-I_{\rm th}^{\rm J}}{I_{\rm th}^{\rm J}},{\overline{I}^{\rm J}(\lambda)>{I_{\rm th}^{\rm J}}},\\
k_{\rm jam}(\lambda),\qquad\qquad\quad{\overline{I}^{\rm J}(\lambda)\leq{I_{\rm th}^{\rm J}}}.\\
\end{array} \right.
\end{equation}

\begin{algorithm}
\caption{Dynamic Collaboration Approach}
\label{alg4}
\begin{algorithmic}[4]
\FOR{$t=1$ to $\Lambda T$}
\STATE Target association according to Algorithm 1.
\STATE Plan the trajectory for UAVs using Algorithm 2.
\IF {${\forall b\!\in\![1,M]},{d(r_{b}(t),X_{a}^{\rm tar}(t))\!=\!\!\mathop{\min}\limits_{a\in[1,M]}d(r_{b}(t),X_{a}^{\rm tar}(t))}$}
\IF {$t=\lambda T$}
\STATE Get the optimal transmit power according to Algorithm 3.
\STATE Obtain the average interference $\overline{I}^{\rm U}(\lambda)$ and the average jamming $\overline{I}^{\rm J}(\lambda)$ according to (\ref{average interference}).
\IF {$\overline{I}^{\rm U}(\lambda)>{I_{\rm th}^{\rm U}}$}
\STATE {\em Trajectory planning is triggered}.
\STATE Update $k_{\rm rep}(\lambda)$ according to (\ref{rep variation}) and replan the trajectory for UAVs using Algorithm 2.
\ENDIF
\IF {$\overline{I}^{\rm J}(\lambda)>{I_{\rm th}^{\rm J}}$}
\STATE {\em Trajectory planning is triggered}.
\STATE Update $k_{\rm jam}(\lambda)$ according to (\ref{jam variation}) and replan the trajectory for UAVs using Algorithm 2.
\ENDIF
\ENDIF
\ELSE
\STATE {\em Target association is triggered}.
\STATE Target association according to Algorithm 1.
\STATE Repeat steps 4 to 14.
\ENDIF
\ENDFOR
\end{algorithmic}
\end{algorithm}

\section{Simulation Results}\label{section4}
The proposed algorithms for target association, trajectory planning, and power control are validated by simulation in this section. In the MTT scenario, $100$ UAVs in the swarm track $5$ targets with $5$ malicious jammers\cite{50}, and these UAVs are distributed randomly in a certain area initially. The simulation parameters are listed in Table~\ref{simulation_parameters}.
\begin{table}[t]
\caption{Simulation Parameters}
\label{simulation_parameters}
\centering
\begin{tabular}{cc|cc}
\hline
Parameters & Value & Parameters & Value\\
\hline
$k_{\rm att}$ & 5 & $V_{\rm max}$ & 20m/s\cite{28}\\
$k_{\rm rep}$, $k_{\rm obs}$, $k_{\rm jam}$ & 10 & $\sigma^2$ & $10^{-8}\rm W$\cite{27}\\
$k_{\rm ext}$ & 5 & $\gamma_{\rm th}$ & $3\rm dB$\cite{27}\\
$d_0$ & 5 & $e(0)$ & $1\rm J$\\
$\widetilde{L}$ & 0.5 & $p_{\rm max}^{\rm U}$ & $0.1\rm W$\cite{02}\\
$\Lambda$ & 15 & $p_{\rm max}^{\rm J}$ & $0.4\rm W$\cite{03}\\
$\alpha$ & $4$ & $I_{\rm th}^{\rm U}$ , $I_{\rm th}^{\rm J}$ & $10^{-3}$\rm W\\
\hline
\end{tabular}
\end{table}
\subsection{Performance of CETA Algorithm}\label{section4A}
In this simulation, we compare the proposed CETA algorithm with the k-means and the fuzzy c-means (FCM) in \cite{04}. The following two evaluation metrics are used to measure the performance of the sub-swarm division.

1) Sum of the Squared Errors (SSE) is the dividing error of all UAVs, which measures the looseness of sub-swarm. A smaller SSE indicates a more compact sub-swarm. SSE is defined as
\begin{equation}\label{SSE}
SSE=\sum\limits_{m=1}^{M}\sum\limits_{i\in C_{m}}\left \| X_{i}-r_{m}\right \|^2.
\end{equation}

2) Silhouette Coefficient (SC) describes the similarity between a sub-swarm where the UAV is located and other sub-swarms. The range is from $-1$ to $+1$. The larger the SC, the better the sub-swarm division. When SC is $1$, it means that the UAV is far from the other sub-swarms. When SC is $0$, it indicates that the UAV may be on the boundary of two sub-swarms. When SC is negative, it implies that the UAV may be misclassified.
The SC of the UAV$_{i}$ is defined as
\begin{equation}\label{SC of each UAV}
SC(i)=\frac{b(i)-a(i)}{{\rm max}\{a(i),b(i)\}},
\end{equation}
where $a(i)=\frac{1}{\mid\Phi_{m}-1\mid}\sum\limits_{i,j\in C_{m},i\neq j}d(X_{i},X_{j})$ is the intra-sub-swarm dissimilarity, denoting the average distance between UAV$_{i}$ and other UAVs in the same sub-swarm. $b(i)=\min\limits_{m^{'}\neq m}\frac{1}{\Phi_{m^{'}}}\sum\limits_{j\in C_{m^{'}}}d(X_{i},X_{j})$ is the inter-sub-swarm dissimilarity, denoting the minimum of the average distance between UAV$_{i}$ and all UAVs in the different sub-swarms. For the UAV swarm, its SC is the average of the SCs of all UAVs, i.e., $SC=\frac{1}{N}\sum\limits_{i=1}^{N}SC(i)$.

\begin{table}
\caption{Results of Three Algorithms}
\label{results of three algorithms}
\centering
\resizebox{.99\columnwidth}{!}{
\begin{tabular}{cccc}
\hline
Clustering algorithms & Number of iterations & SSE$\ast(10^3)$ & SC\\
\hline
CETA & $1$ & $9.1589$ & $0.5966$\\
k-means & $19$ & $9.5229$ & $0.5711$\\
FCM & $70$ & $9.6223$ & $0.5461$\\
\hline
\end{tabular}
}
\end{table}
\begin{figure}
  \centering
  \includegraphics[width=8cm]{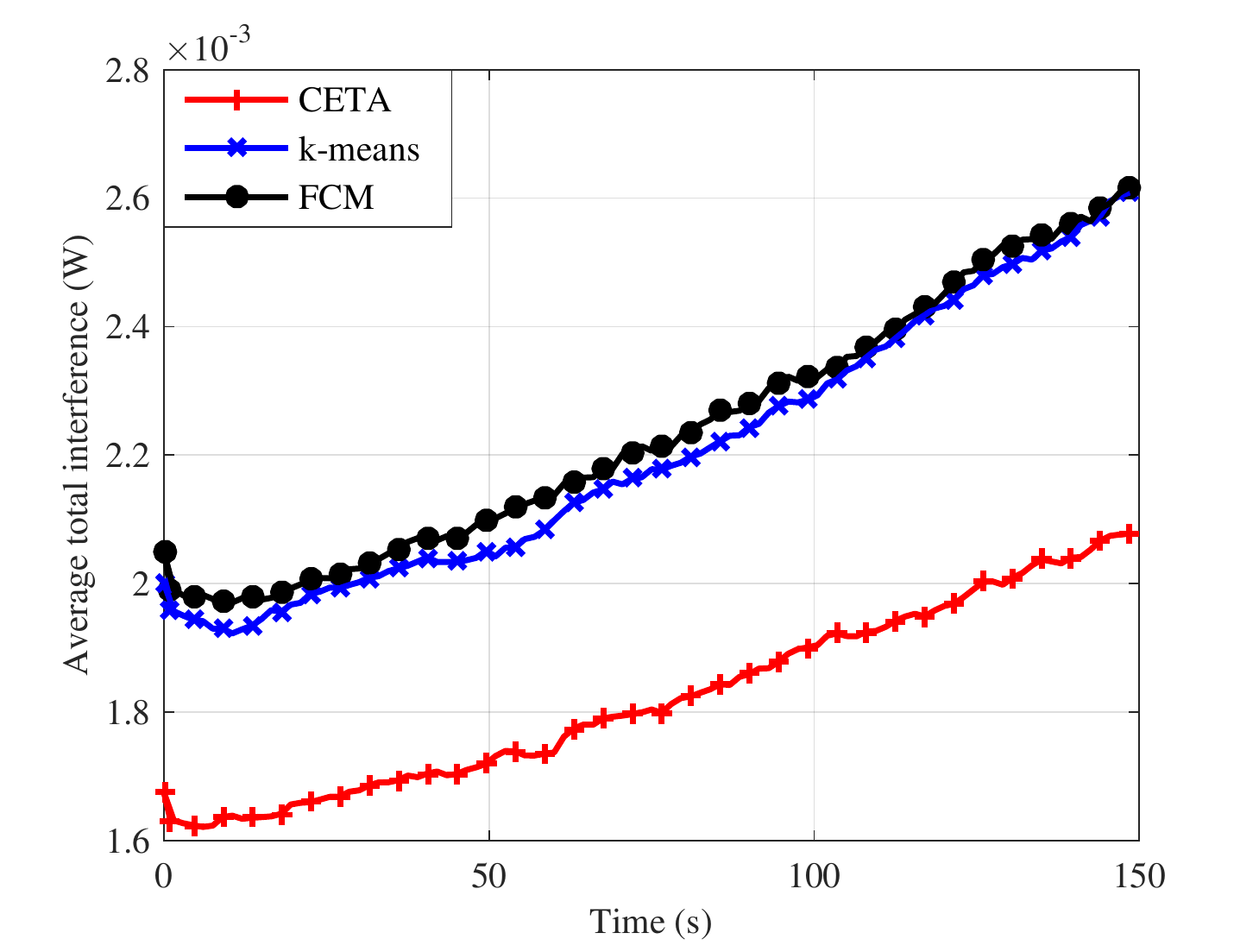}
  \caption{Average total interference of a UAV with different clustering algorithms.}
  \label{fig5}
\end{figure}
In simulations, we perform $1,000$ experiments to compare the performance of the CETA algorithm with benchmarks. The results of $1,000$ experiments are averaged, and shown in Table~\ref{results of three algorithms}. We can see that the proposed CETA algorithm has an excellent advantage in the number of iterations. For the performance of SSE and SC, our proposed algorithm outperforms k-means and FCM, and indicating that the proposed method can be more reasonably divided into multiple sub-swarms.
\begin{figure}
  \centering
  \includegraphics[width=8cm]{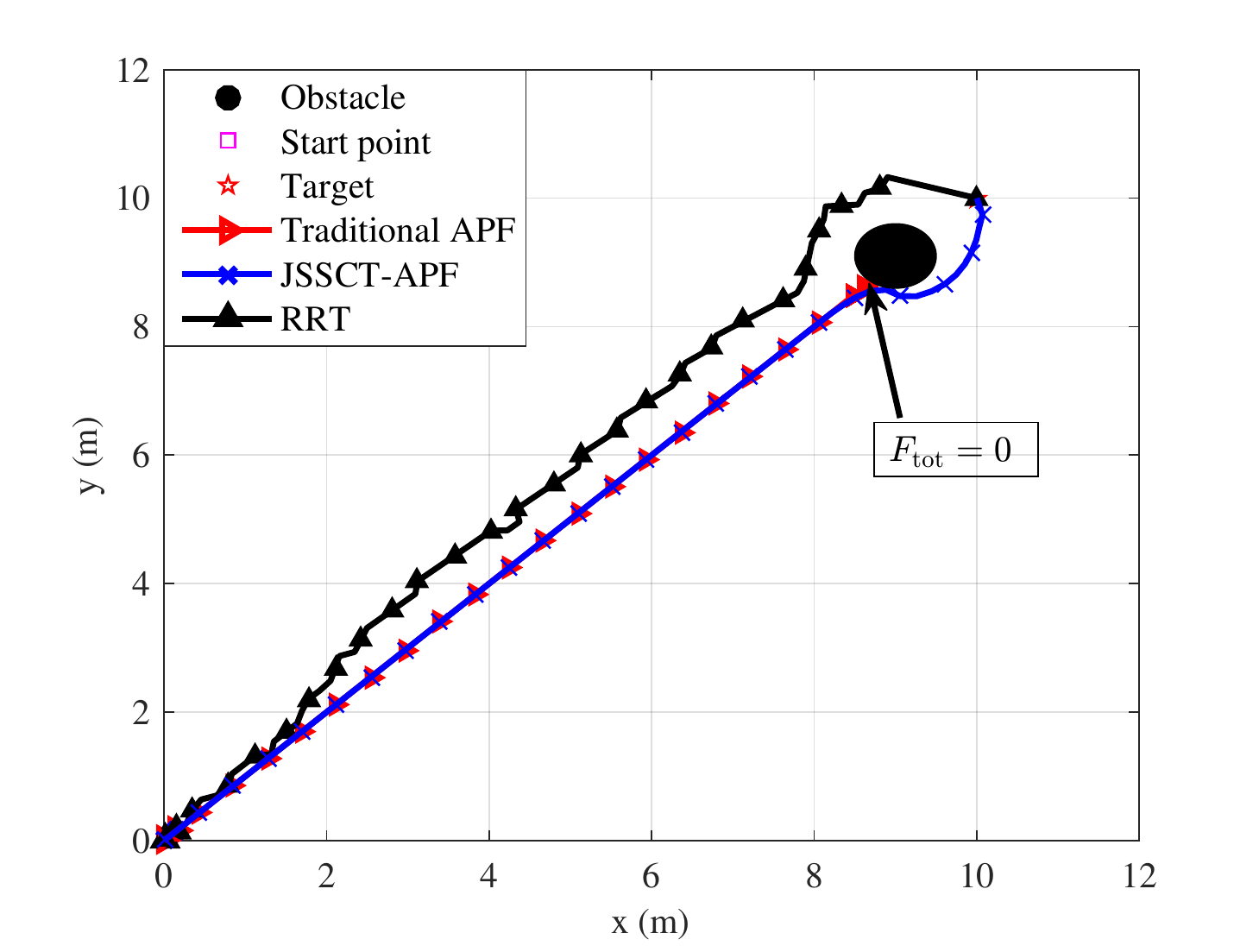}
  \caption{The trajectories of UAVs under three different trajectory planning algorithms when the UAV passes the balance point. The starting point is located at $(0,0)$, and the target is located at $(10,10)$.}
  \label{fig13a}
\end{figure}

Fig. \ref{fig5} exhibits the average total interference of a UAV during the tracking when leveraging CETA algorithm, k-means, and FCM in the dynamic collaboration approach, respectively. We can observe that the UAV's average total interference generally shows an upward trend over time because the distance among UAVs and between UAV and jammer decreases with time. Note that some small fluctuations in the interference are caused by the irregularity of the UAV's channel gain. Fig. \ref{fig5} indicates that our proposed CETA algorithm has better performance in alleviating interference. It is because that the sub-swarm division obtained by CETA algorithm is more compact than that obtained by the k-means and FCM as shown in Table~\ref{results of three algorithms}. As a result, the UAVs are mainly influenced by intra-subswarm interference with the CETA algorithm, and the UAVs are influenced by both intra-subswarm and inter-subswarm interference with the k-means and FCM.
\subsection{Performance of JSSCT-APF}\label{section4B}
\begin{figure}
  \centering
  \includegraphics[width=8cm]{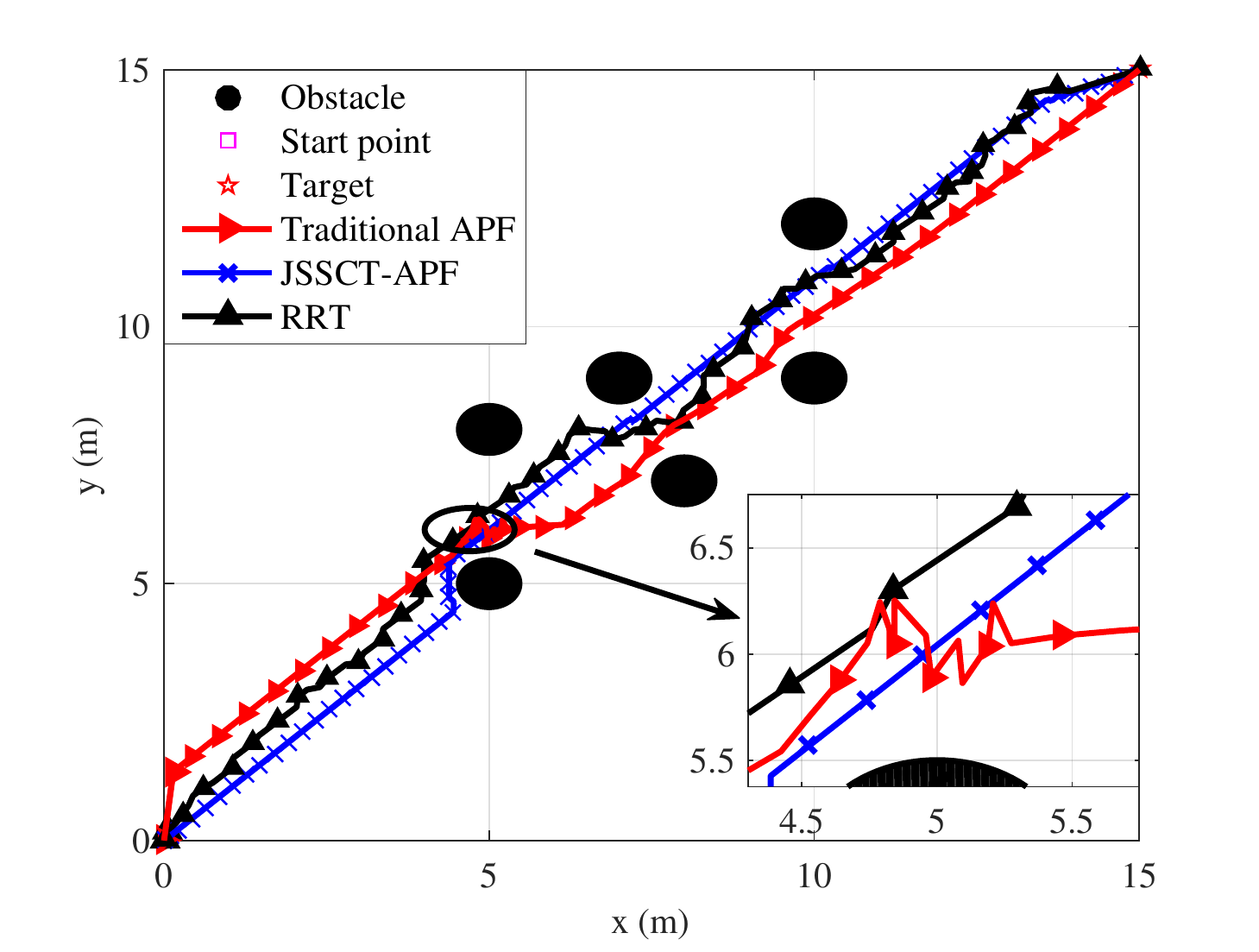}
  \caption{The UAV trajectories under three different trajectory planning algorithms when the UAV jitters. The starting point is located at $(0,0)$, and the target is located at $(15,15)$.}
  \label{fig13b}
\end{figure}

\begin{figure}
  \centering
  \includegraphics[width=8cm]{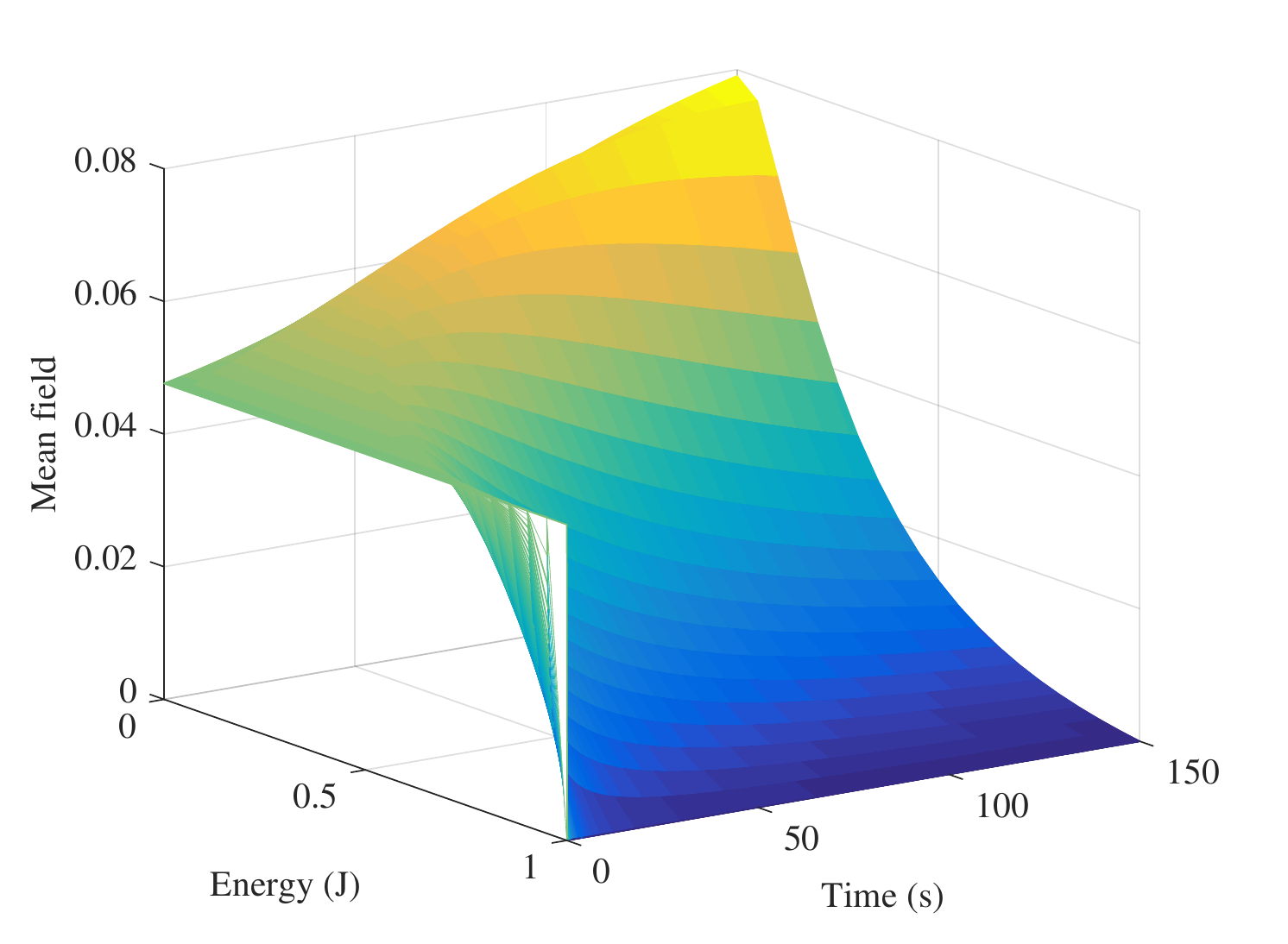}
  \caption{Mean field distribution w.r.t. time and energy state at MFE.}
  \label{fig6}
\end{figure}
\begin{figure}
  \centering
  \includegraphics[width=8cm]{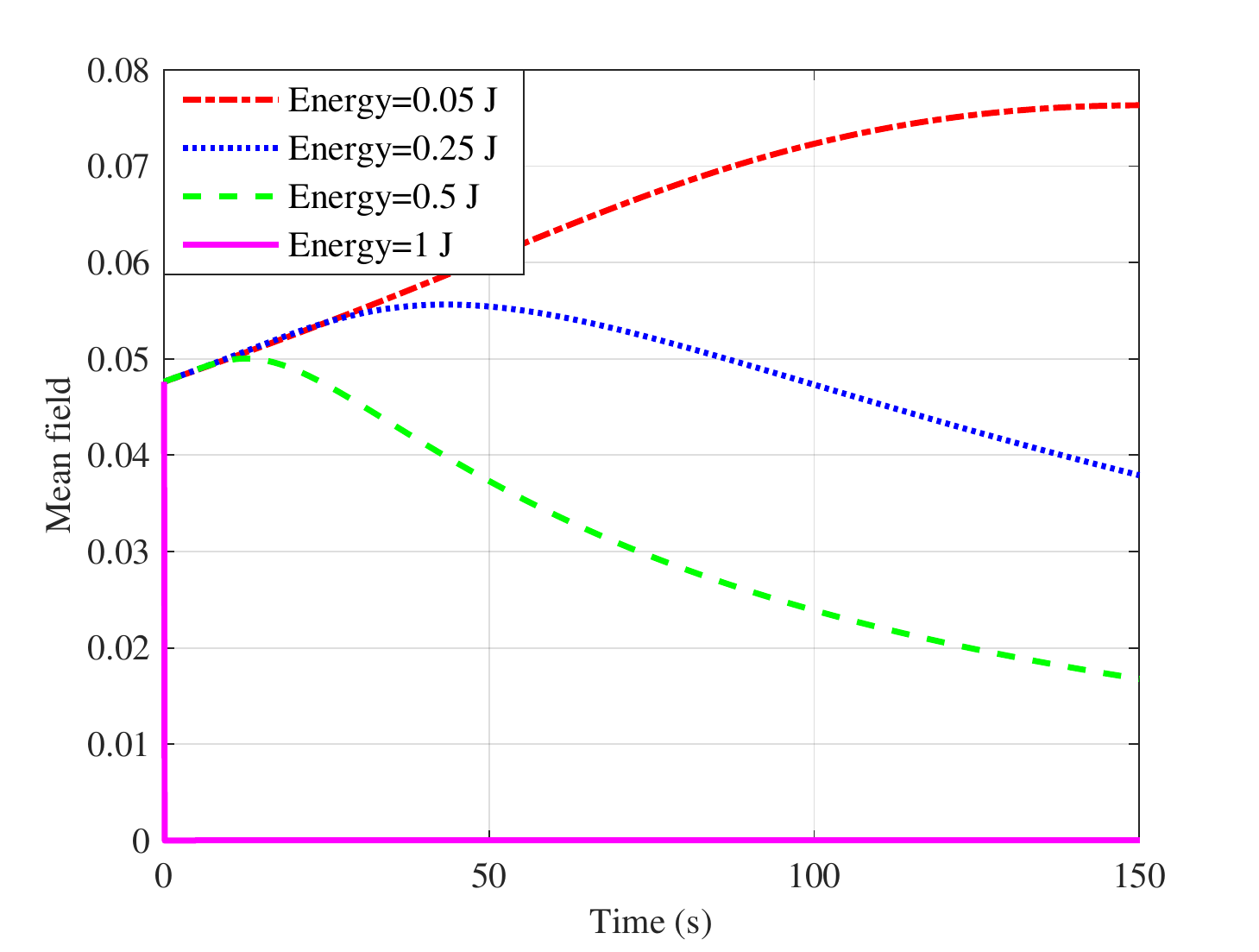}
  \caption{Mean field distribution under some specific energy levels with time.}
  \label{fig7}
\end{figure}
Fig. \ref{fig13a} and Fig. \ref{fig13b} show the trajectory of a UAV under three algorithms of the proposed JSSCT-APF, the traditional APF and the rapid-exploration random tree (RRT) in \cite{47} in two different cases. Since the settings of the parameters related to APFs in the existing literatures are different and only need to satisfy greater than zero, this paper refers to the settings of \cite{06}. As shown in Fig. \ref{fig13a}, when the UAV passes the balance point, it parks at this point and is unable to reach its target with the traditional APF while it successfully tracks the target with the JSSCT-APF and RRT. In Fig. \ref{fig13b}, the UAV jitters with the traditional APF when it flies through a narrow space. However, the JSSCT-APF and RRT prevent the UAV from jitter. The reason is that the JSSCT-APF can by correcting the turning angle erasure the jitter which is the inherent defect of the traditional APF. Furthermore, it can be seen that the trajectory with the JSSCT-APF is shorter and more smooth than that with the RRT. This is because the direction of path generation with the RRT has strong randomness, while the JSSCT-APF is always oriented towards the target. Therefore, it can be concluded that the proposed JSSCT-APF excels the traditional APF and the RRT in trajectory planning.

\subsection{Performance of JA-MFG}\label{section4C}
Fig. \ref{fig6} and Fig. \ref{fig8} show mean field distribution and the power control policy at MFE, respectively. In the simulations, the initial energy distribution $M(0,:)$ follows a uniform distribution~\cite{19,21}.

Fig. \ref{fig6} shows the mean field distribution w.r.t. time and energy state at MFE. As it is shown, the probability of UAVs at the zero energy level tends to increase over time, and the probability of UAVs at higher energy levels tends to decrease over time. The probability of UAVs at higher energy levels but not at the full energy level, does not drops to $0$ with time eventually. This means that not all UAVs run out of their energy. This is because the UAV will not increase its transmit power after meeting the minimum SNR requirements according to (\ref{optimal control policy}).

\begin{figure}
  \centering
  \includegraphics[width=8cm]{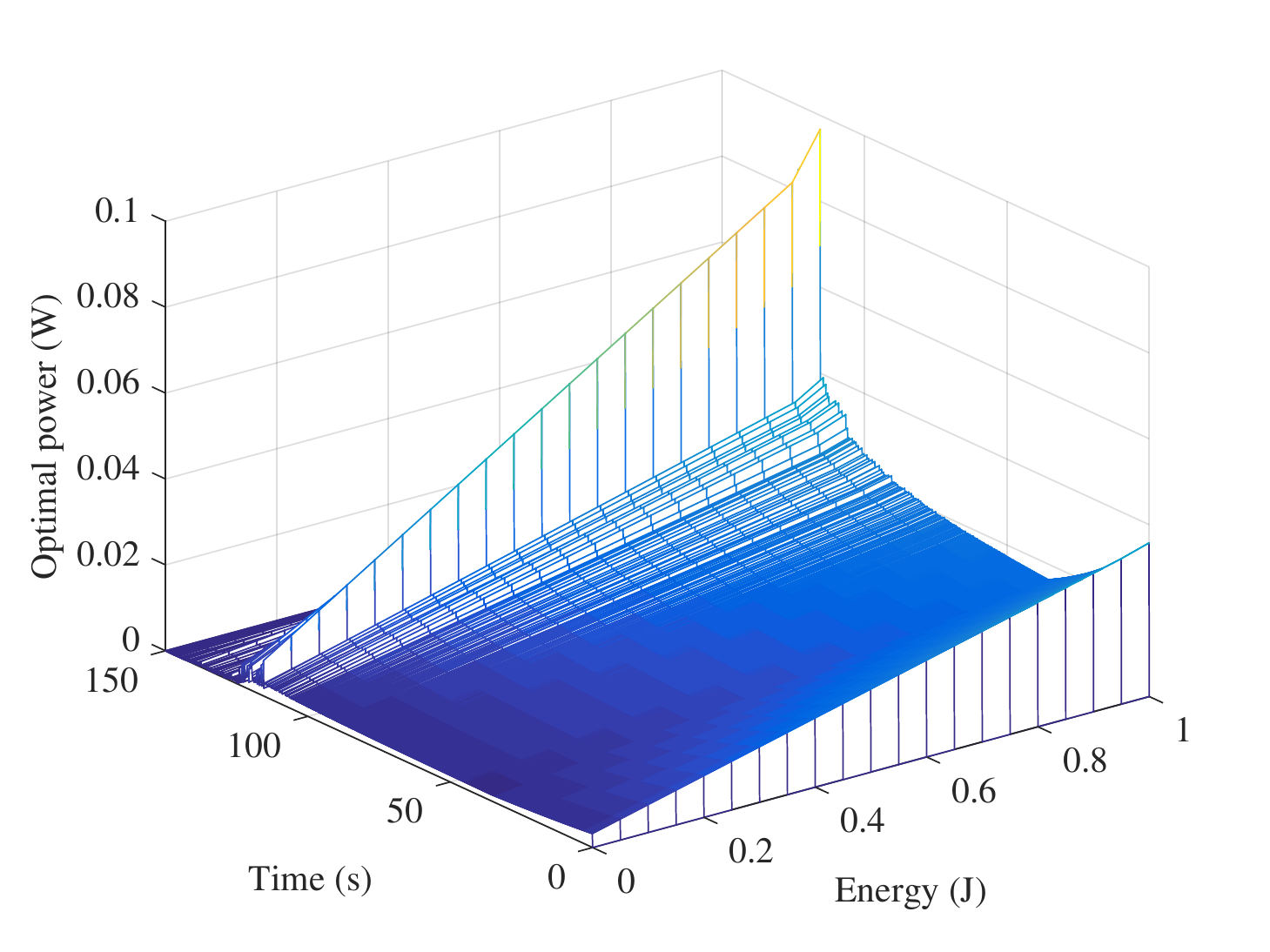}
  \caption{Optimal transmit power w.r.t. time and energy state at MFE.}
  \label{fig8}
\end{figure}
Fig. \ref{fig7} depicts several cross-sections of Fig. \ref{fig6} to better present the changes of the mean field distribution of UAVs at some specific energy levels with time. The probability of the UAV in all energy levels is equal when $t=0$ because of the uniform initial energy distribution. The probability of UAVs with $0.05 \rm J$ rises rapidly at first, and then remains almost unchanged. The probability of UAVs at full energy level immediately falls to zero at the beginning of the tracking. The reason is that the UAVs have to consume power to drive SINR as close to $\gamma_{\rm th}$ as possible.

Fig. \ref{fig8} illustrates the optimal power policy w.r.t. time and energy state at MFE. As it is illustrated, UAVs are constantly updating the transmit power to attain MFE. The UAVs with high energy level have a higher transmit power than ones with lower energy level.
\begin{figure}
\centering

\subfigure[]{
\begin{minipage}{0.49\linewidth}
\centering
\includegraphics[width=1.1\linewidth]{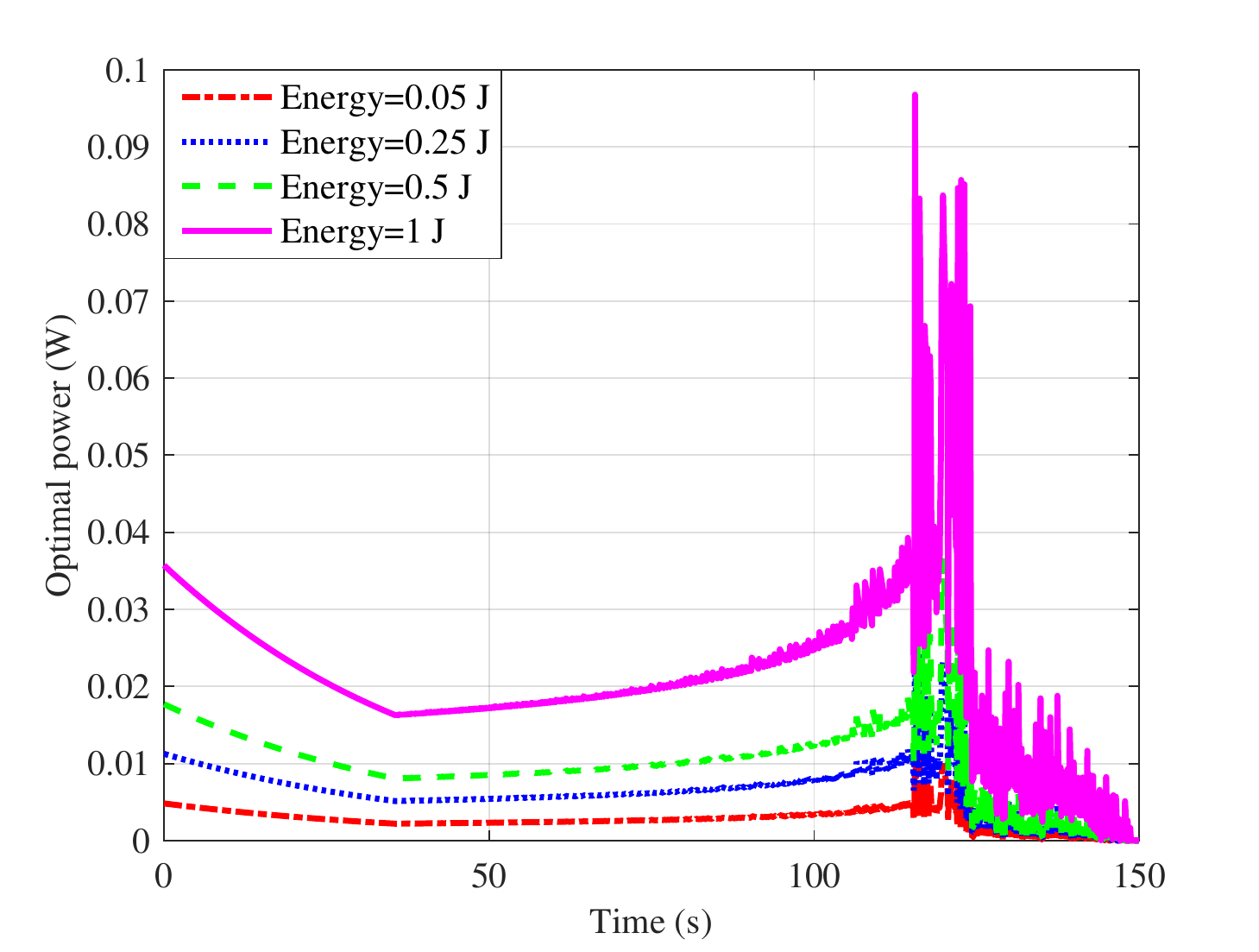}
\label{fig9a}
\end{minipage}%
}
\subfigure[]{
\begin{minipage}{0.49\linewidth}
\centering
\includegraphics[width=1.1\linewidth]{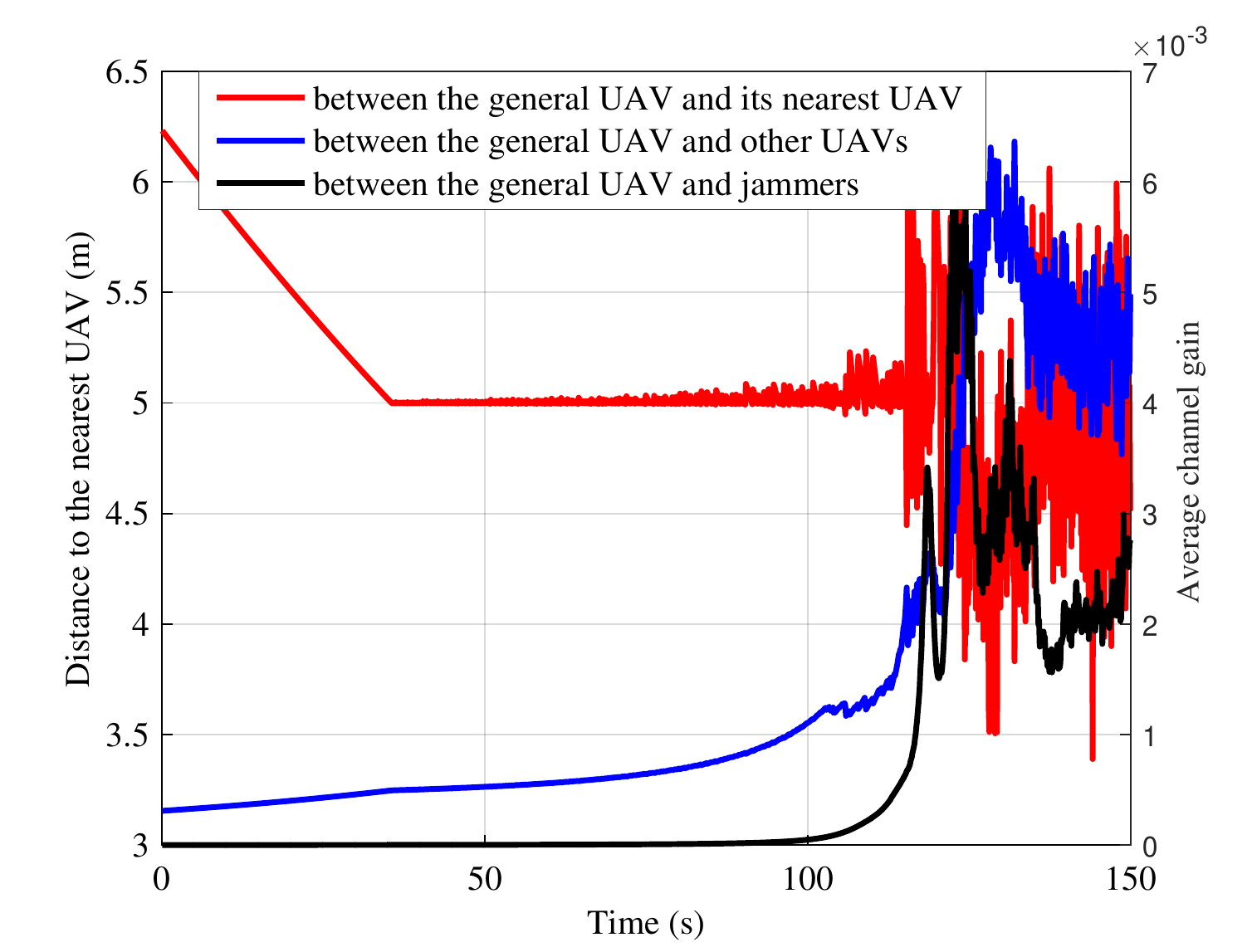}
\label{fig9b}
\end{minipage}
}%
\quad
\centering
\caption{(a) Optimal transmit power under some specific energy levels with time. (b) The solid red line reflects the variation of distance between a UAV and its nearest UAV with time. The solid blue line reflects the variation of average channel gain between a UAV and other UAVs with time. The solid black line reflects the variation of average channel gain between a UAV and jammers with time.}
\label{fig9}
\end{figure}

\begin{figure}
\vspace{-2.0em}
\centering

\subfigure[]{
\begin{minipage}[t]{8cm}
\centering
\includegraphics[width=8cm]{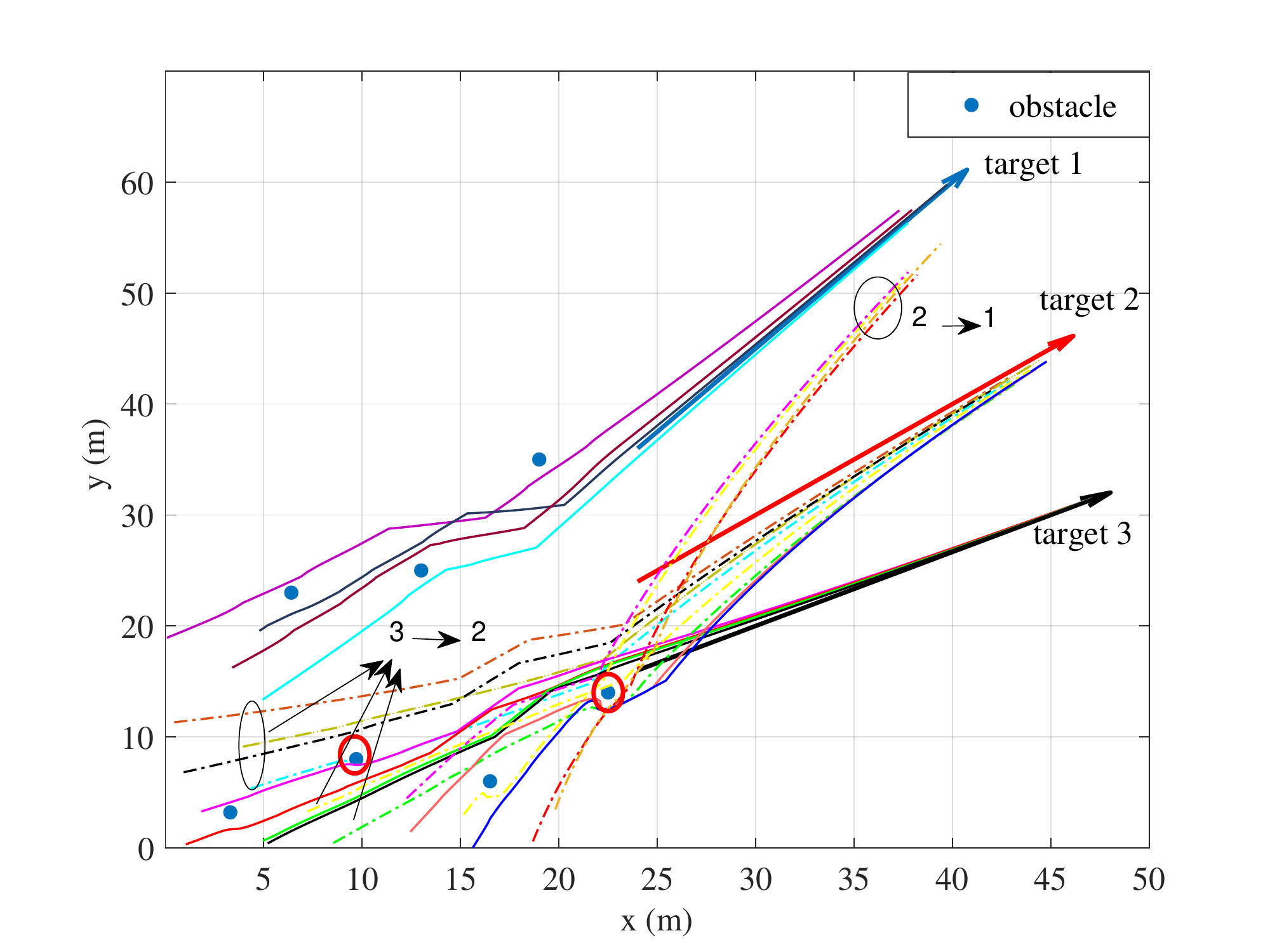}
\label{fig3a}
\end{minipage}%
}%
\quad                 
\subfigure[]{
\begin{minipage}[t]{8cm}
\centering
\includegraphics[width=8cm]{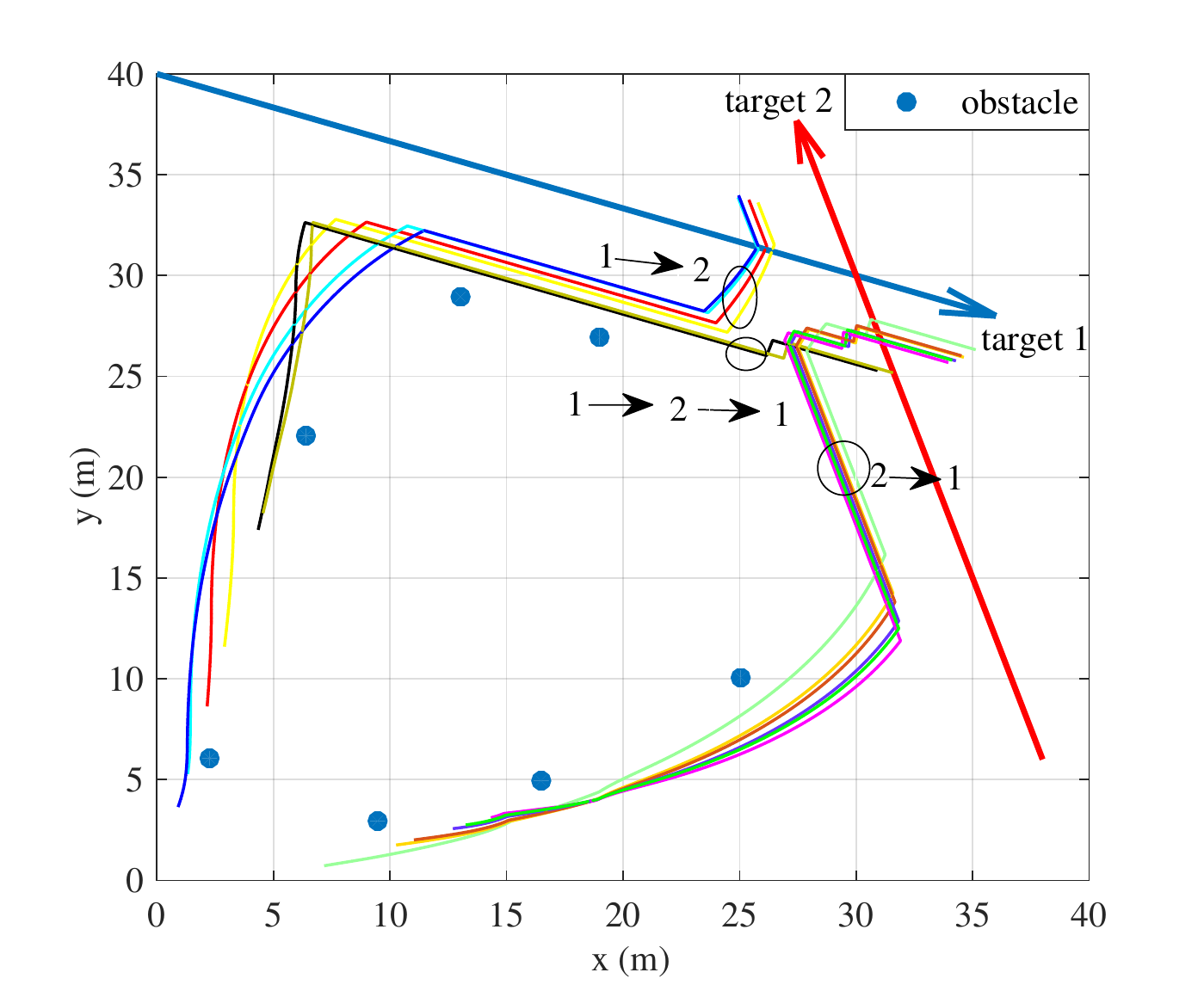}
\label{fig3b}
\end{minipage}
}%
\quad
\centering
\caption{The UAV trajectories under the combination of the CETA algorithm and the JSSCT-APF in two different cases: (a) The three targets fly in different directions. (b) The two targets fly towards each other. The blue dots are obstacles. The straight lines with an arrow are the trajectories of targets. The dashed and solid curves in different colors represent the UAVs' trajectories. $1$ $\rightarrow$ $2$ represents UAVs' tracking target switching from $1$ to $2$.}
\label{fig3}
\end{figure}
Fig. \ref{fig9a} depicts several cross-sections of Fig. \ref{fig8} to better present the changes of the optimal power at some specific energy levels with time. As shown in the figure, the optimal power fluctuates greatly with time, especially at the end of tracking. According to (\ref{optimal power}), the optimal power $\widetilde p^{\rm U}$ is determined by the channel gain between a UAV and its nearest UAV $g_{i,i^{'}}^{\rm U}$, the interference mean field $\overline{I}_{i^{'}}^{\rm U}$, and the jamming mean field $\overline{I}_{i^{'}}^{\rm J}$. Furthermore, we can observe that $g_{i,j}^{\rm U}$, $\overline{I}_{i^{'}}^{\rm U}$, and $\overline{I}_{i^{'}}^{\rm J}$ are correlated with the distance $d(X_{i},X_{i^{'}})$ between a UAV and its nearest UAV, the average interference channel gain $\overline{g}^{\rm U}$, and the average jamming channel gain $ \overline{g}^{\rm J}$, respectively. Therefore, it can be concluded that $d(X_{i},X_{i^{'}})$, $\overline{g}^{\rm U}$, and $ \overline{g}^{\rm J}$ are crucial for $\widetilde p^{\rm U}$. Their variation with time are shown in Fig. \ref{fig9b}, which are used together to explain the fluctuation of the optimal power.
%
\subsection{Performance of Dynamic Collaboration Approach}\label{section4D}
To keep the UAV tracking the nearest target, we propose a criterion for timely target reassociation. As discussed before, when the UAV's current tracking target is not the closest one to the centroid of its corresponding sub-swarm, the target is reassociated. To show the effectiveness of this criterion, the proposed CETA and JSSCT-APF are combined without considering communication and jamming in Fig. \ref{fig3}. It is observed that the UAVs replace their tracking target in time and reach all the targets smoothly\footnote{Note that UAVs do not collide with the obstacle (circled in red), but fly over the outside of the obstacle in Fig. \ref{fig3a}.} in different cases.
\begin{figure}
\vspace{-2.0em}
\centering

\subfigure[]{
\begin{minipage}[t]{8cm}
\centering
\includegraphics[width=8cm]{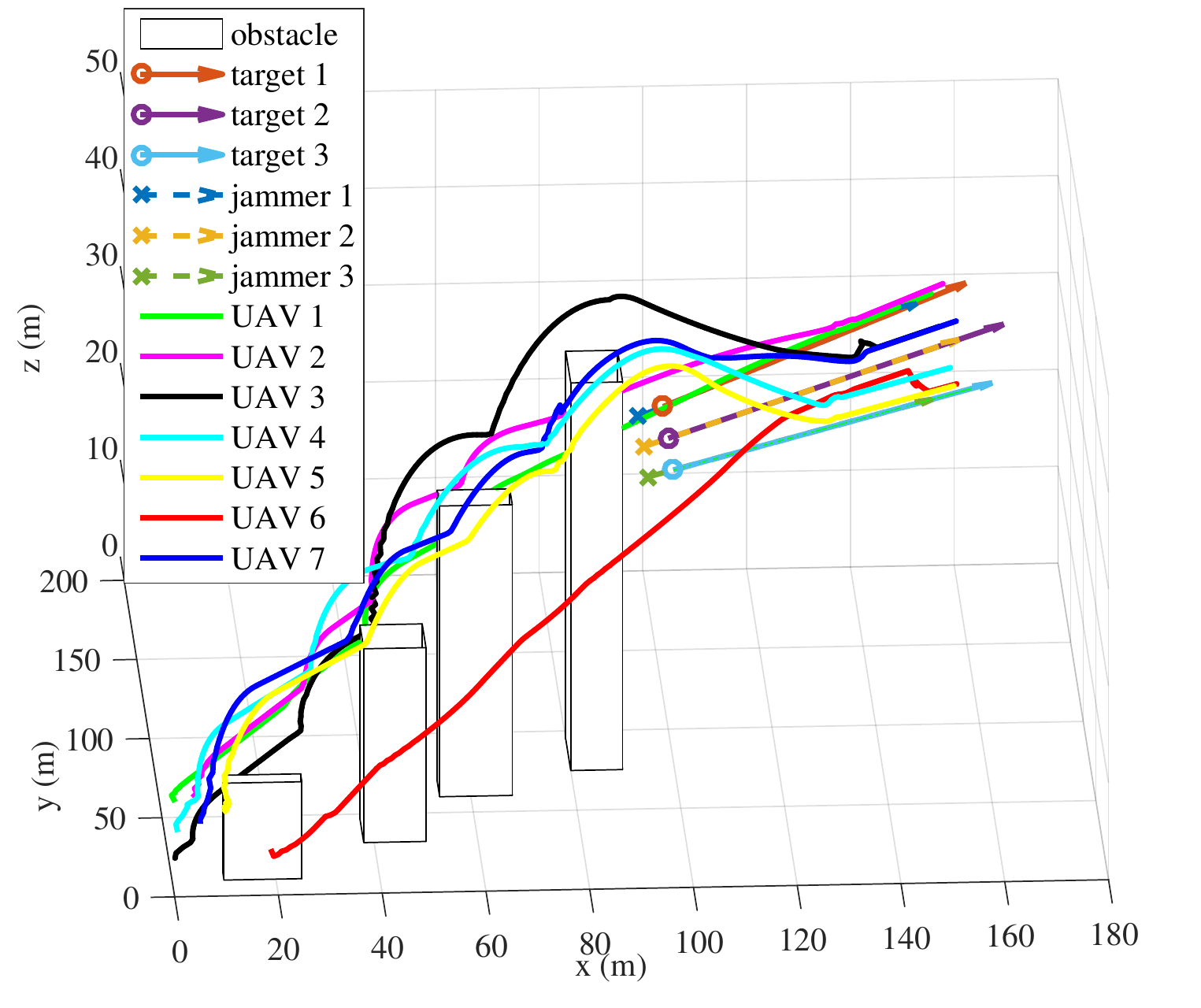}
\label{fig4a}
\end{minipage}%
}%
\quad                 
\subfigure[]{
\begin{minipage}[t]{8cm}
\centering
\includegraphics[width=8cm]{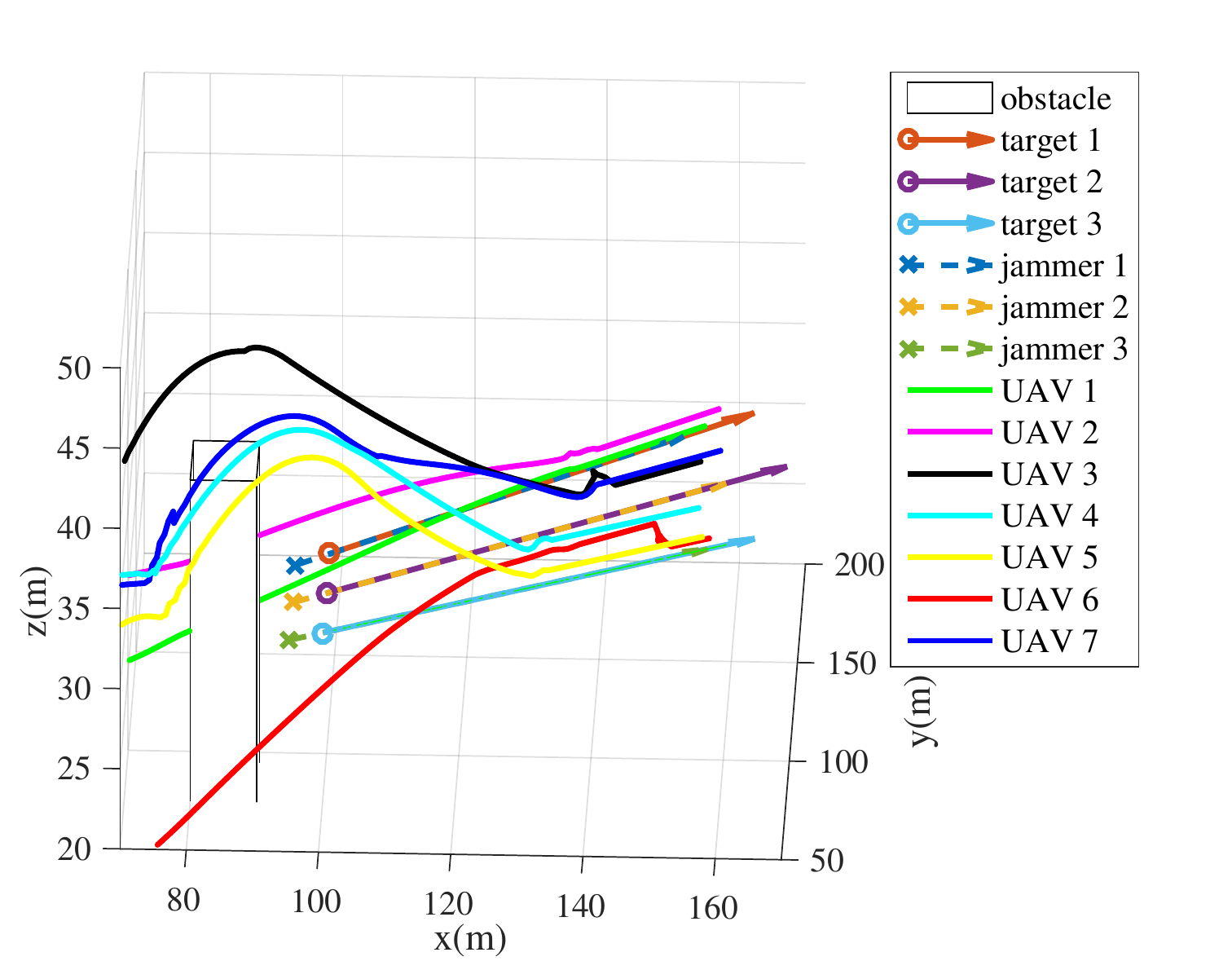}
\label{fig4b}
\end{minipage}
}%
\quad
\centering
\caption{The trajectories of seven UAVs selected from swarm with the CETA algorithm combining the JSSCT-APF. The cuboids are obstacles. The circles represent the starting position of the targets. The x-mark represents the starting position of the jammers. (a) Complete map of UAVs', targets', and jammers' trajectory. (b) Partial enlarged detail of UAVs', targets', and jammers' trajectory in (a).}
\label{fig4}
\end{figure}

Considering the difficulty of clearly showing the trajectory of each UAV owing to the numerous UAVs in the swarm, the trajectories of seven UAVs selected from swarm are shown in Fig. \ref{fig4}. We can observe that UAVs reach the target successfully and are free of the collision with obstacles, other UAVs in swarm, and jammers. Meanwhile, several UAVs change targets during tracking. For example, UAV $3$ tracks target $1$ initially and then target $2$, and UAV $6$ changes its tracking target from target $2$ to target $3$. In order to more clearly exhibit the positional relationship between the target and the jammer, we give a partially enlarged view of the trajectory in Fig. \ref{fig4b}.
\begin{figure}
  \centering
  \includegraphics[width=8cm]{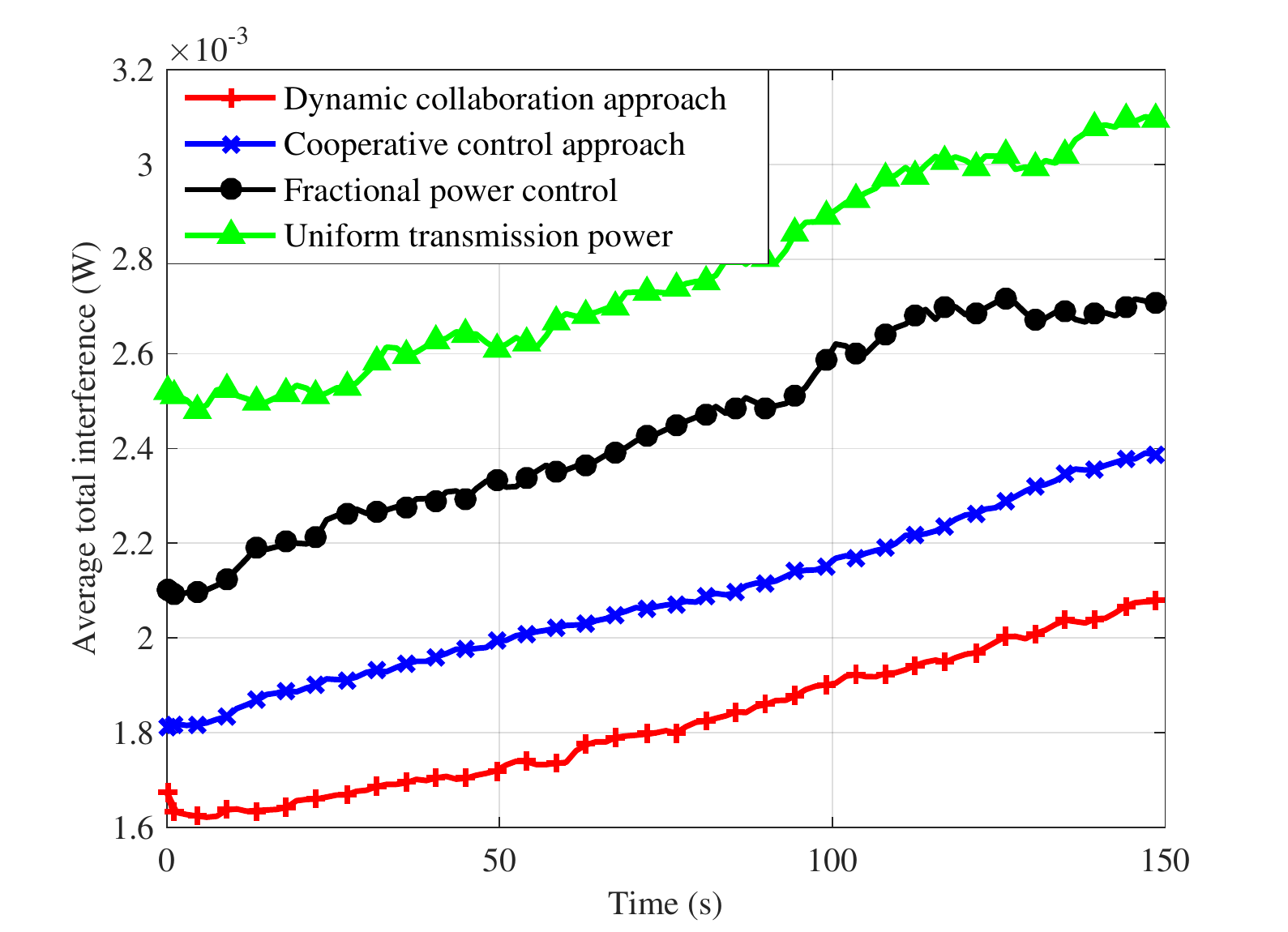}
  \caption{Average total interference of a UAV with different power control algorithms. Cooperative control approach in \cite{06} refers to cooperative control enabled by SCT-APF and MFG without considering jammer, fractional power control in \cite{05} refers to cooperative control of SCT-APF and fractional power control without considering jammer, and uniform transmission power refers to cooperative control of SCT-APF and uniform transmission power without considering jammer.}
  \label{fig10}
\end{figure}

\begin{figure}
  \centering
  \includegraphics[width=8cm]{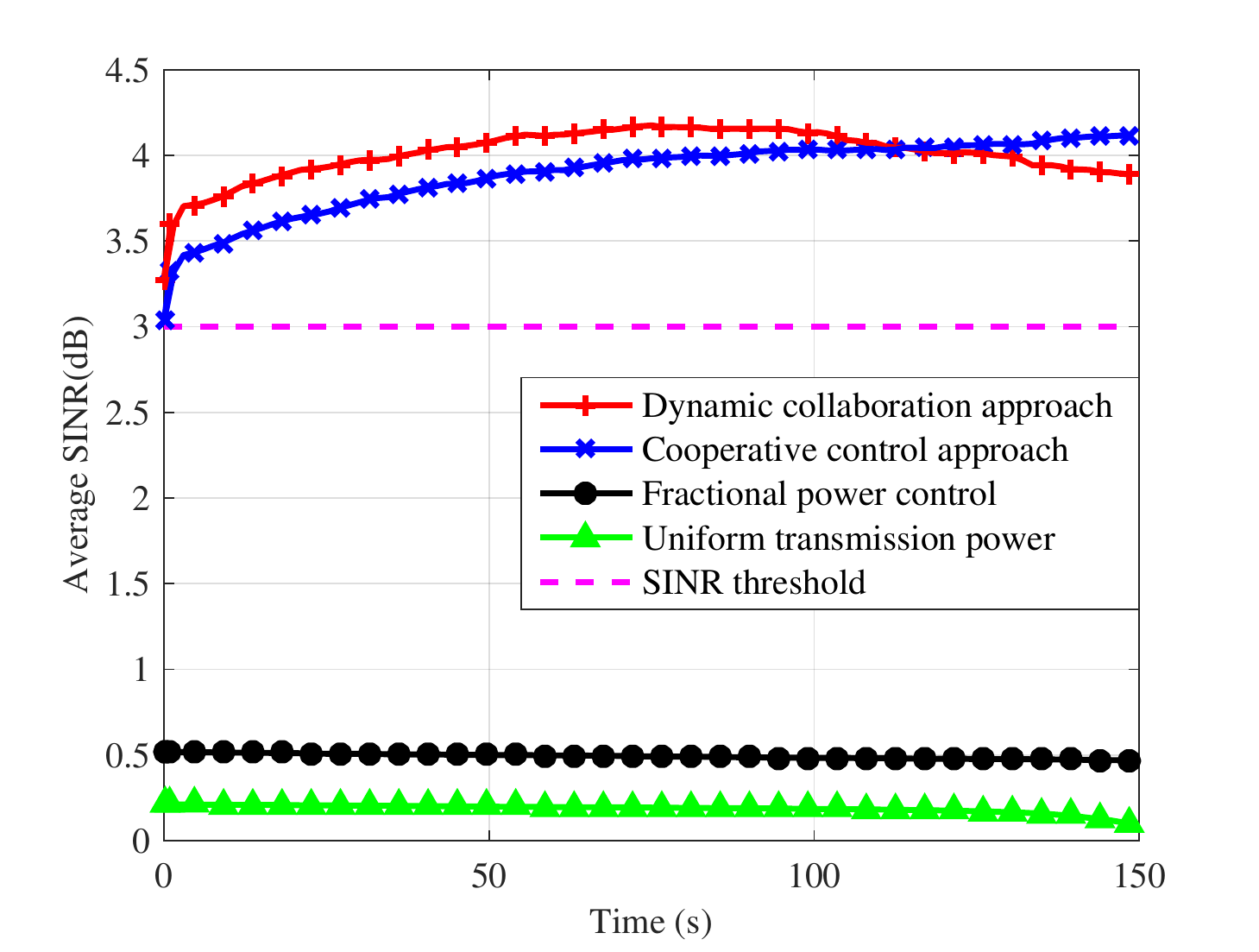}
  \caption{Average SINR of a UAV with different power control algorithms.}
  \label{fig11}
\end{figure}
To assess the performance of the dynamic collaboration approach under malicious jamming, we compare our proposed  algorithm with cooperative control approach in \cite{06}, cooperative control of SCT-APF and fractional power control \cite{05}, and cooperative control of SCT-APF and uniform transmit power policy (${\forall t}\in[0, \lambda T]$, the transmit power $p=\frac{e(0)}{T}$).

Fig. \ref{fig10} reveals the average total interference of a UAV with time under various power control approaches. Similar to Fig. \ref{fig5}, the UAV's average total interference shows an upward trend over time, and our proposed dynamic collaboration approach outperforms benchmarks under malicious jamming. This is because that our proposed dynamic collaboration approach takes the collision with jammers and their jamming interference into consideration, which can relieve the interference effectively in both trajectory and power.

Fig. \ref{fig11} illustrates the average SINR of a UAV with time under various power control approaches. We can see that the UAV's average SINR is larger than the SINR threshold all the time with our proposed dynamic collaboration approach and cooperative control approach in \cite{06}. While for fractional power control in \cite{05} and uniform transmit power policy, the requirement of minimum SINR can not be fulfilled. Besides, the average SINR under our proposed dynamic collaboration approach firstly is larger than that under cooperative control approach in \cite{06}, and then smaller than that under cooperative control approach in \cite{06}. The reason is that the inter-UAV interference and jamming have a sharp increase at the end of tracking as shown in Fig. \ref{fig9b}, and the UAV has to sacrifice the performance of SINR to offset the increased interference to minimize cost function with our proposed dynamic collaboration approach.
\begin{table}
\caption{Performance of Four Algorithms}
\label{results of four algorithms}
\centering
\resizebox{.99\columnwidth}{!}{
\begin{tabular}{ccc}
\hline
Algorithm & Number of steps & {\makecell[c]{Average number of \\target switching}}\\
\hline
dynamic collaboration approach & $1481$ & $2.12$\\
\hline
{\makecell[c]{cooperative control approach in \cite{06}}} & $1746$ & $3.46$\\
\hline
{\makecell[c]{cooperative control of SCT-APF \\and fractional power control}} & $1943$ & $3.9$\\
\hline
{\makecell[c]{cooperative control of SCT-APF \\and uniform transmission power policy}}& $2214$ & $4.1$\\
\hline
\end{tabular}
}
\end{table}

To illustrate the effect of interference on tracking quality, we compared the number of steps and the average number of target switching required by UAVs to track all targets under the above four algorithms in Table~\ref{results of four algorithms}. The results indicate that the dynamic collaboration approach has fewer the number of steps and the average number of target switching compared to benchmarks. Thus, our proposed dynamic collaboration approach with less interference contributes to improving the tracking quality.

\section{Conclusion}\label{section5}
This paper studies target association, trajectory planning, and power control for UAV swarm under malicious jamming, where many UAVs are employed to simultaneously track multiple targets. We deploy the CETA algorithm, JSSCT-APF, and JA-MFG power control scheme to divide UAVs into sub-swarms with the associated target, plan trajectory, and control power, respectively. To minimize the total interference, a dynamic collaboration approach is presented to re-associate the target and replan the trajectory. The experiment results verify that the proposed algorithm helps UAVs reach all targets successfully and ensures communication quality between UAVs. Moreover, the proposed dynamic collaboration algorithm achieves dynamic adjustment of target association, trajectory, and power. Meanwhile, the proposed algorithm considering jammer alleviates the interference by $28\%$, and enhances the tracking quality in tracking steps and target switching times by $33\%$ and $48\%$, respectively, compared to conventional algorithms without considering jammer.

%
\appendices
\section{Proof of Proposition 1}
When each UAV and jammer transmit a predefined test power $p_{1}^{\rm U}(t)$ and $p_{1}^{\rm J}(t)$, the received power at UAV$_{i^{'}}$ would be
\begin{equation}\label{first received power}
\begin{aligned}
p_{1}^{\rm R}(t)&=p_{1}^{\rm U}(t)g_{i,i^{'}}^{\rm U}(t)\!+\!\!\sum\limits_{n\neq{i,i^{'}}}^{N}p_{1}^{\rm U}(t)g_{n,i^{'}}^{\rm U}(t)\!+\!\!\sum\limits_{m=1}^{M}p_{1}^{\rm J}(t)g_{m,i^{'}}^{\rm J}(t)\\
&\approx p_{1}^{\rm U}(t)g_{i,i^{'}}^{\rm U}(t)+(N-2)p_{1}^{\rm U}(t)\overline{g}^{\rm U}(t)+Mp_{1}^{\rm J}(t)\overline{g}^{\rm J}(t).
\end{aligned}
\end{equation}

Similarly, when the predefined test power are $p_{2}^{\rm U}(t)$ and $p_{2}^{\rm J}(t)$, we have
\begin{equation}\label{second received power}
p_{2}^{\rm R}(t)\approx p_{2}^{\rm U}(t)g_{i,i^{'}}^{\rm U}(t)+(N-2)p_{2}^{\rm U}(t)\overline{g}^{\rm U}(t)+Mp_{2}^{\rm J}(t)\overline{g}^{\rm J}(t).
\end{equation}

Solving above two equations simultaneously, we can obtain the average interference channel gain $\overline{g}^{\rm U}$ and the average jamming channel gain $\overline{g}^{\rm J}$ as follows
\begin{equation}\label{average interference channel gain}
\overline{g}^{\rm U}(t)=\frac{\frac{p_{1}^{\rm R}(t)p_{2}^{\rm J}(t)-p_{2}^{\rm R}(t)p_{1}^{\rm J}(t)}{p_{1}^{\rm U}(t)p_{2}^{\rm J}(t)-p_{2}^{\rm U}(t)p_{1}^{\rm J}(t)}-g_{i,i^{'}}^{\rm U}(t)}{N-2},
\end{equation}
\begin{equation}\label{average jamming channel gain}
\overline{g}^{\rm J}(t)=\frac{p_{1}^{\rm R}(t)p_{2}^{\rm U}(t)-p_{2}^{\rm R}(t)p_{1}^{\rm U}(t)}{M(p_{1}^{\rm J}(t)p_{2}^{\rm U}(t)-p_{2}^{\rm J}(t)p_{1}^{\rm U}(t))},
\end{equation}
where $p_{(\cdot)}^{\rm R}(t)$ are measured, and $g_{i,i^{'}}^{\rm U}(t)$ is previously known.

Therefore, we can derive the interference mean field and jamming mean field as shown in (\ref{detail interference mean field}) and (\ref{detail jamming mean field}).
\section{Proof of Proposition 2}
Because of the identity and interchangeability of all agents, the subscript $i$ can be omitted in JA-MFG. Therefore, Hamiltionian in (\ref{Hamiltionian}) can be rewritten as
\begin{equation}\label{rewritten Hamiltionian}
\begin{aligned}
H\left(e(t),\frac{\partial{u(t,e(t))}}{\partial{e}}\right)\!=\!\mathop{\min}_{p^{\rm U}(t)}\left(c(p^{\rm U}(t))-p^{\rm U}(t)\frac{\partial{u(t,e(t))}}{\partial{e}}\right).
\end{aligned}
\end{equation}
Substituting (\ref{mean field cost function}) into (\ref{rewritten Hamiltionian}), we can obtain
\begin{equation}\label{detail rewritten Hamiltionian}
\begin{aligned}
H&=\mathop{\min}_{p^{\rm U}(t)}\left(\omega_1[p^{\rm U}(t)g_{i,i^{'}}^{\rm U}(t)-\gamma_{\rm th}(\overline{I}_{i^{'}}^{\rm U}(t)+\overline{I}_{i^{'}}^{\rm J}(t)+\sigma^2(t))]^{2}\right.\\
&\phantom{=\;\;}\left.+\omega_2(\overline{I}_{i^{'}}^{\rm U}(t)+\overline{I}_{i^{'}}^{\rm J}(t))^{2}-p^{\rm U}(t)\frac{\partial{u(t,e(t))}}{\partial{e}}\right).
\end{aligned}
\end{equation}
Calculating the first derivative of (\ref{detail rewritten Hamiltionian}) w.r.t. $p^{\rm U}$ yields
\begin{equation}\label{first derivatives of detail rewritten Hamiltionian}
\begin{aligned}
\frac{\partial{H}}{\partial{p^{\rm U}}}\!&=\!2\omega_1(p^{\rm U}(t)g_{i,j}^{\rm U}(t)\!-\!\gamma_{\rm th}(\overline{I}_{i^{'}}^{\rm U}(t)+\overline{I}_{i^{'}}^{\rm J}(t)\!+\!\sigma(t)^2))g_{i,j}^{\rm U}(t)\\
&-\frac{\partial{u(t,e(t))}}{\partial{e}}.
\end{aligned}
\end{equation}
Let $\frac{\partial{H}}{\partial{p^{\rm U}}}=0$, and the optimal power in (\ref{optimal power}) can be derived.

\ifCLASSOPTIONcaptionsoff
  \newpage
\fi



%


\bibliographystyle{IEEEtran}
\bibliography{bibi}

%







\end{document}